\documentclass[aap,MSNbibl,seceqn,dvips]{arximspdf}
\usepackage{graphicx}
\doi{10.1214/13-AAP957} 
\volume{24}
\issue{4}
\pubyear{2014}
\firstpage{1585}
\lastpage{1620}

\makeatletter
\newcommand{\widebar}{\overline}
\newcommand{\rrVert}{\Vert}
\newcommand{\rrvert}{\vert}
\newcommand{\llVert}{\Vert}
\newcommand{\llvert}{\vert}
\newtheorem{theorem}{Theorem}[section]
\newtheorem{lemma}[theorem]{Lemma}
\newtheorem{coro}[theorem]{Corollary}
\newproclaim{assp}[theorem]{Assumption}
\def\X{\widebar{X}}
\newcommand{\halfs}{1/2}
\makeatother

\begin{document}
\begin{frontmatter}

\title{Antithetic multilevel Monte Carlo estimation
for~multi-dimensional SDEs without\\ L\'evy area simulation}
\pdftitle{Antithetic multilevel Monte Carlo estimation
for multi-dimensional SDEs without Levy area simulation}
\runtitle{Antithetic multilevel Monte Carlo}

\begin{aug}
\author{\fnms{Michael B.} \snm{Giles}\ead[label=e1]{mike.giles@maths.ox.ac.uk}}
\and
\author{\fnms{Lukasz} \snm{Szpruch}\corref{}\ead[label=e2]{szpruch@maths.ox.ac.uk}}
\runauthor{M.~B. Giles and L. Szpruch}
\affiliation{University of Oxford}
\address{Mathematical Institute\\
University of Oxford\\
Oxford OX1 3LB\\
United Kingdom\\
\printead{e1}\\
\phantom{E-mail:\ }\printead*{e2}} 
\end{aug}

\received{\smonth{2} \syear{2012}}
\revised{\smonth{11} \syear{2012}}

%
\begin{abstract}
In this paper we introduce a new multilevel Monte Carlo
(MLMC) estimator for multi-dimensional SDEs driven by Brownian motions.
Giles has previously shown that if we combine a numerical approximation
with strong order of convergence $O(\Delta t)$ with MLMC we can reduce
the computational complexity to estimate expected values of
functionals of SDE solutions with a root-mean-square error of
$\epsilon$ from $O(\epsilon^{-3})$ to $O(\epsilon^{-2})$. However, in general,
to obtain a rate of strong convergence higher than $O(\Delta t^{1/2})$ requires
simulation, or approximation, of L{\'e}vy areas. In this paper, through
the construction of a suitable antithetic multilevel correction estimator,
we are able to avoid the simulation of L{\'e}vy areas and still achieve an
$O(\Delta t^2)$ multilevel correction variance for smooth payoffs, and almost
an $O(\Delta t^{3/2})$ variance for piecewise smooth payoffs, even
though there
is only $O(\Delta t^{1/2})$ strong convergence. This results in an
$O(\epsilon^{-2})$
complexity for estimating the value of European and Asian put and call options.
\end{abstract}

%
\pdfkeywords{65C30, 65C05, Monte Carlo, multilevel, Levy area, stochastic differential equation}

\end{frontmatter}

\section{Introduction}

In many financial engineering applications, one is interested in the expected
value of a financial derivative whose payoff depends upon the solution
of a stochastic differential equation (SDE). Using a simple Monte Carlo
method with
a numerical discretisation with first order weak convergence, to
achieve a
root-mean-square error of $\epsilon$ would require $O(\epsilon
^{-2})$ independent
paths, each with $O(\epsilon^{-1})$ time steps, giving
a computational complexity which is $O(\epsilon^{-3})$,~\cite{dg95}.

Recently, Giles \cite{giles08} introduced a multilevel Monte Carlo (MLMC)
estimator which enables a reduction of this computational cost to
$O(\epsilon^{-2} (\log{(1/\epsilon)})^2)$ for Lipschitz payoffs
when using the
Euler--Maruyama discretisation. For other discontinuous and path-dependent
payoff functions, the complexity is poorer \cite{ghm09}.
The efficiency of the MLMC method is influenced by the strong
convergence order of the discretisation, and subsequent research
using MLMC with the first-order Milstein discretisation for scalar SDEs,
improved the complexity significantly to $O(\epsilon^{-2})$ for
digital, lookback and barrier options \cite{giles08b}.
However, a weakness of the Milstein discretisation is that in multiple
dimensions it generally requires the simulation of iterated It{\^o}
integrals known as L{\'e}vy areas, for which
there is no known efficient method except in dimension 2
\cite{gl94,rw01,wiktorsson01}.

Let $(\Omega,{\mathcal{F}},\{{\mathcal{F}}_t\}_{t\geq0},\mathbb{P})$
be a complete probability space with a filtration
$\{{\mathcal{F}}_t\}_{t\geq0}$ satisfying the usual conditions,
and let
$w(t)$
be a $D$-dimensional
Brownian motion defined on the probability space. We consider the
numerical approximation of a~general class of multi-dimensional SDEs
driven by Brownian of the form
%
\begin{equation}
\label{eqSDE} \mathrm{d}x(t)=f\bigl(x(t)\bigr)\,\mathrm{d}t+g\bigl(x(t)\bigr) \,
\mathrm{d}w(t),
\end{equation}
where $x(t)\in\mathbb{R}^{d}$ for each $t \ge 0$,
$f\in C^2 (\mathbb{R}^{d}, \mathbb{R}^{d})$,
$g\in C^2 (\mathbb{R}^{d}, \mathbb{R}^{d\times D})$,
and for simplicity we assume a fixed initial value
$x_{0}\in\mathbb{R}^{d}$.

In this paper we are primarily concerned with estimating $\mathbb
{E}[P(x(T))]$, the
expected value of a payoff depending on the solution at a fixed time $T$,
defining the tensor $h_{ijk}(x)$ as
%
\begin{equation}
\qquad h_{ijk}(x) = \frac{1}{2} \sum_{l=1}^d
g_{lk}(x) \frac{\partial g_{ij}}{\partial x_l}(x), \qquad i=1,\ldots ,d\mbox{ and } k,j=1,\ldots,D, \label{eqh}
\end{equation}
when using $N$ uniform timesteps $\Delta t = T/N$,
the $i$th component of the first order Milstein approximation
$\widehat{X}_n \approx x(n \Delta t)$ has the form \cite{kp92}
%
\begin{eqnarray}\label{eqMilstein}
\widehat{X}_{i,n+1} &=& \widehat{X}_{i,n} +
f_i(\widehat{X}_n) \Delta t + \sum
_{j=1}^D g_{ij}(\widehat{X}_n)
\Delta w_{j,n}
\nonumber\\[-8pt]\\[-8pt]
&&{} + \sum_{j,k=1}^D h_{ijk}(\widehat{X}_n) (\Delta w_{j,n}\Delta
w_{k,n} - \Omega_{jk} \Delta t - A_{jk,n} ),\nonumber
\end{eqnarray}
where
$\Omega$ is the correlation matrix for the driving Brownian paths,
and $A_{jk,n}$ is the L{\'e}vy area defined as
\[
A_{jk,n} = \int_{t_n}^{t_{n+1}} \bigl(w_j(t) - w_j(t_n) \bigr) \,
\mathrm{d} w_k(t) - \int_{t_n}^{t_{n+1}}
\bigl(w_k(t) - w_k(t_n)
\bigr) \,\mathrm{d} w_j(t).
\]
In some applications, the diffusion coefficient $g(x)$ has a commutativity
property which gives
$h_{ijk}(x) = h_{ikj}(x)$ for all $i,j,k$.
In that case, because the L{\'e}vy areas are anti-symmetric
(i.e., $A_{jk,n}=-A_{kj,n}$), it follows that
$h_{ijk}(X_n) A_{jk,n} + h_{ikj}(X_n) A_{kj,n} = 0$
and therefore the terms involving the L{\'e}vy areas cancel and so it is not
necessary to simulate them. However, this only happens in special cases.


Clark and Cameron \cite{cc80} proved for a particular SDE that it is
impossible to achieve a better order of strong convergence than the
Euler--Maruyama discretisation when using just the discrete increments
of the underlying Brownian motion. The analysis was extended by
M\"{u}ller--Gronbach \cite{muller02} to general SDEs.
As a consequence if we use the standard MLMC method with the Milstein
scheme \mbox{without} simulating the L{\'e}vy areas the complexity will remain
the same as for Euler--Maruyama.
Nevertheless, in this paper we show
that by constructing a suitable antithetic MLMC estimator one can
neglect the
L{\'e}vy areas and still obtain a~multilevel correction estimator with a
variance which decays at the same rate as the scalar Milstein estimator.
This demonstrates that a high order of the strong convergence is not necessary
for our new estimator to achieve the optimal complexity $O(\epsilon^{-2})$.


We begin the paper by reviewing the multilevel Monte Carlo approach,
introducing the idea of the antithetic estimator and bounding the
behaviour of its variance under certain conditions.
Because of its simplicity, we then consider Clark and Cameron's model
problem, and prove that the antithetic path simulations do satisfy the
required conditions to give an $O(\Delta t^2)$ variance convergence
for a
smooth payoff.
This then motivates the subsequent analysis for the general class
of multi-dimensional SDEs.
We support our analysis by suitable numerical experiments in which
we demonstrate the superiority of antithetic MLMC over the standard MLMC
for both the Clark--Cameron SDE and the Heston stochastic volatility model.
The \hyperref[app]{Appendix} contains the detailed proofs of the key theorems.

In this paper we restrict attention to financial applications
with either a European payoff, dependent on the final value $x(T)$,
or an Asian payoff, dependent on the average of $x(t)$ over the time
interval $[0,T]$.
It is proved that when the payoff is twice differentiable, with
bounded derivatives, the rate of convergence of the multilevel
correction variance is doubled from $O(\Delta t)$ to $O(\Delta t^2)$.
If the payoff is Lipschitz, and twice differentiable almost everywhere,
then the rate of convergence is reduced to $O(\Delta t^{3/2})$, but
this is
still sufficient to make the overall complexity $O(\epsilon^{-2})$ to achieve
a root-mean-square accuracy of $\epsilon$.

\section{Multilevel Monte Carlo estimation} \label{secMLMC}

\subsection{MLMC estimators}

In its most general form, multilevel Monte Carlo simulation uses a
number of
levels of resolution, $\ell=0,1,\ldots,L$, with $\ell=0$ being the coarsest,
and $\ell=L$ being the finest. In the context of a SDEs simulation,
level~$0$
may have just one timestep for the whole time interval $[0,T]$, whereas level
$L$ might have $2^L$ uniform timesteps.

If $P$
denotes the payoff (or other output functional of interest), and
$P_{\ell}$
denote its approximation on level $l$, then
the expected value $\mathbb{E}[P_L]$ on the finest level is equal to
the expected value $\mathbb{E}[P_0]$
on the coarsest level plus a sum of corrections which give the
difference in
expectation between simulations on successive levels,
%
\begin{equation}
\label{eqMLMC} \mathbb{E}[P_{L}] = \mathbb{E}[P_{0}] + \sum
_{\ell=1}^{L}\mathbb {E}[P_{\ell}-P_{\ell-1}].
\end{equation}
The idea behind MLMC is to independently estimate each of the expectations
on the right-hand side of (\ref{eqMLMC}) in a way which minimises the overall
variance for a given computational cost. Let $Y_0$ be an estimator for
$\mathbb{E}[P_{0}]$ using $N_{0}$ samples, and let $Y_\ell$, $\ell
>0$, be an
estimator for $\mathbb{E}[P_{\ell} - P_{\ell-1}]$ using $N_{\ell}$
samples. The simplest
estimator is a mean of $N_{\ell}$ independent samples, which for $\ell
>0$ is
%
\begin{equation}
\label{eqest} Y_{\ell} = N_{\ell}^{-1}\sum
_{i=1}^{N_{\ell}} \bigl(P^{i}_{\ell
}-P^{i}_{\ell-1}
\bigr).
\end{equation}
The key point here is that $P^{i}_{\ell} - P^{i}_{\ell-1}$ should come
from two discrete approximations for the same underlying stochastic sample,
so that on finer levels of resolution the difference is small (due to strong
convergence) and so the variance is also small. Hence very few samples
will be
required on finer levels to accurately estimate the expected value.

Here we recall the theorem from \cite{giles2012stochastic} (which is a
slight generalisation of
the original theorem in \cite{giles08}) which gives the complexity of
MLMC estimation.

%
\begin{theorem} \label{thcomplexity}
Let $P$ denote a functional of the solution of a stochastic differential
equation, and let $P_\ell$ denote the corresponding level $\ell$ numerical
approximation. If there exist independent estimators $Y_\ell$ based on
$N_\ell$ Monte Carlo samples, and positive constants
$\alpha, \beta, \gamma, c_1, c_2, c_3$ such that
$\alpha\geq\frac{1}{2} \min(\beta,\gamma)$ and:
\begin{longlist}[(iii)]
\item[(i)]
$\llvert  \mathbb{E}[P_\ell- P] \rrvert  \leq c_1 2^{-\alpha
\ell}$,\vspace*{5pt}

\item[(ii)]
$\displaystyle
\mathbb{E}[Y_\ell] = \cases{
\mathbb{E}[P_0], &\quad$\ell=0$, \vspace*{2pt}\cr
\mathbb{E}[P_\ell- P_{\ell-1}], &\quad$\ell>0$,}$\vspace*{5pt}

\item[(iii)]
$\mathbb{V}[Y_\ell] \leq c_2 N_\ell^{-1} 2^{-\beta \ell}$,\vspace*{5pt}

\item[(iv)]
$C_\ell\leq c_3 N_\ell 2^{\gamma \ell}$,
where $C_\ell$ is the computational complexity of $Y_\ell$,
\end{longlist}
then there exists a positive constant $c_4$ such that for any
$\epsilon < e^{-1}$
there are values $L$ and $N_\ell$ for which the multilevel estimator
\[
Y = \sum_{\ell=0}^L Y_\ell
\]
has a mean-square-error with bound
\[
\mbox{MSE} \equiv\mathbb{E} \bigl[ \bigl(Y - \mathbb{E}[P] \bigr)^2 \bigr]
< \epsilon^2
\]
with a computational complexity $C$ with bound
\[
C \leq\cases{ c_4 \epsilon^{-2}, &\quad$\beta>\gamma$,
\vspace*{2pt}\cr
c_4 \epsilon^{-2} \bigl(\log{(1/\epsilon)}
\bigr)^2, &\quad$\beta=\gamma$,
\vspace*{2pt}\cr
c_4
\epsilon^{-2-(\gamma- \beta)/\alpha}, &\quad$0<\beta<\gamma$.}
\]
\end{theorem}

Without the simulation of L{\'e}vy areas, the strong order of convergence
of the Milstein discretisation $X(T)$ which is used is only $1/2$, so that
\[
\mathbb{E} \bigl[ \bigl\| x(T) - X(T) \bigr\|^2 \bigr] = O(\Delta t).
\]
Hence, for payoffs which are a Lipschitz function of the final value,
it follows that
\[
\mathbb{E} \bigl[ (P_\ell- P_{\ell-1})^2 \bigr] =
O(\Delta t)
\]
and therefore the estimator given by (\ref{eqest}) satisfies condition
(iii) in the theorem with $\beta= 1$ when $\Delta t \propto
2^{-\ell}$.
What we will show is that \textit{without improving the strong order of
convergence}
it is possible to construct an antithetic estimator for which $\beta=2$.

To do so, we need to exploit some flexibility in the construction of the
multilevel estimator.
In (\ref{eqest}) we have used the same estimator for the payoff
$P_{\ell}$ on every level $\ell$, and therefore (\ref{eqMLMC}) is a
trivial identity due to the telescoping summation. However, in
\cite{giles08b} Giles numerically showed that it can be better to use different
estimators for the finer and coarser of the two levels being
considered, $P^{f}_{\ell}$ when level $\ell$ is the finer level,
and $P^{c}_{\ell}$ when level $\ell$ is the coarser level.
In this case, we require that
%
\begin{equation}
\label{conMLMC} \mathbb{E}\bigl[P^{f}_{\ell}\bigr] = \mathbb{E}
\bigl[P^{c}_{\ell}\bigr]\qquad\mbox{for } \ell=1,\ldots,L,
\end{equation}
so that
\[
E \bigl[P^f_{L}\bigr] = \mathbb{E}\bigl[P^f_{0}
\bigr] + \sum_{\ell=1}^{L}\mathbb {E}
\bigl[P^{f}_{\ell}-P^{c}_{\ell-1}\bigr].
\]

The MLMC theorem is still applicable to this modified estimator. The
advantage is that it gives the flexibility to construct approximations for
which $P^{f}_{\ell}-P^{c}_{\ell-1}$ is much smaller than the original
$P_{\ell}-P_{\ell-1}$, giving a larger value for $\beta$, the rate
of variance
convergence in condition (iii) in the theorem.

\subsection{Antithetic MLMC estimator}

Based on the well-known method of antithetic variates (see, e.g.,
\cite{glasserman04}), the idea for the antithetic estimator is to
exploit the flexibility of the more general MLMC estimator\vspace*{-2pt} by defining
$P^{c}_{\ell-1}$ to be the usual payoff $P(X^c)$ coming from a level
$\ell- 1$ coarse simulation $X^c$, and define $P^{f}_{\ell}$
to be the average of the payoffs $P(X^{f}),P(X^{a})$ coming from an
antithetic pair of level $\ell$ simulations, $X^{f}$ and $X^{a}$.

$X^{f}$ will be defined in a way which corresponds naturally to the
construction of $X^{c}$. Its antithetic ``twin''
$X^{a}$ will be defined so that it has exactly the same distribution
as $X^{f}$, conditional on $X^c$, which ensures that
$
\mathbb{E}[ P(X^{f}) ] = \mathbb{E}[P(X^{a})]
$
and hence (\ref{conMLMC}) is satisfied, but at the same time
\[
\bigl( X^{f} - X^c \bigr) \approx- \bigl(
X^{a} - X^c \bigr)
\]
and therefore
\[
\bigl( P\bigl(X^{f}\bigr) - P\bigl(X^c\bigr) \bigr)
\approx- \bigl( P\bigl(X^{a}\bigr) - P\bigl(X^c\bigr)
\bigr),
\]
so that
$
\frac{1}{2}  ( P(X^{f}) + P(X^{a})  ) \approx P(X^c)$.
This leads to $\frac{1}{2}  ( P(X^{f}) + P(X^{a})  ) - P(X^c)$
having a much smaller variance than the standard estimator $P(X^{f})- P(X^c)$.

We now present a lemma which motivates the rest of the paper by giving
an upper bound on the convergence of the variance of
$\frac{1}{2}  ( P(X^{f}) + P(X^{a})  ) - P(X^c)$.

%
\begin{lemma} \label{lemmapayoff}
If $P\in C^2 (\mathbb{R}^{d}, \mathbb{R})$ and there exist
constants $L_1, L_2$ such that for all $x\in\mathbb{R}^{d}$
\[
\biggl\llVert \frac{\partial P}{\partial x} \biggr\rrVert \le L_1,\qquad \biggl
\llVert \frac{\partial^2 P} {\partial x^2} \biggr\rrVert \le L_2,
\]
then for $p\ge2$,
\begin{eqnarray*}
&& \mathbb{E} \bigl[ \bigl(\tfrac{1}{2} \bigl(P\bigl(X^{f}
\bigr) + P\bigl(X^{a}\bigr)\bigr) - P\bigl(X^c\bigr)
\bigr)^p \bigr]
\\
&&\qquad\le2^{p-1} L_1^p \mathbb{E} \bigl[
\bigl\llVert \tfrac{1}{2}\bigl(X^{f} + X^{a}\bigr) -
X^{c} \bigr\rrVert ^p \bigr] + 2^{-(p+1)}
L_2^p \mathbb{E} \bigl[ \bigl\llVert X^{f} -
X^{a} \bigr\rrVert ^{2p} \bigr].
\end{eqnarray*}
\end{lemma}

\begin{pf}
If we define $\X^f \equiv\frac{1}{2}(X^{f} + X^{a})$, then
a Taylor expansion gives
\[
P\bigl(X^{f}\bigr) = P\bigl(\X^f\bigr) +
\frac{\partial P}{\partial x}^T \bigl(\X^f\bigr) \bigl(X^{f}
- \X^f\bigr) + \frac{1}{2} \bigl(X^{f} -
\X^f\bigr)^T \frac{\partial^2 P} { \partial x^2}(\xi_1)
\bigl(X^{f} - \X^f\bigr)
\]
for some $\xi_1$ on the line between $\X^f$ and $X^{f}$.
Performing a similar expansion for $P(X^{a})$ and then averaging the two,
the linear terms cancel, and one obtains
\begin{eqnarray*}
\frac{1}{2} \bigl(P\bigl(X^{f}\bigr) + P\bigl(X^{a}
\bigr)\bigr) &=& P\bigl(\X^f\bigr) + \frac{1}{4}
\bigl(X^{f} - \X^f\bigr)^T \frac{\partial^2 P} { \partial x^2}(
\xi_1) \bigl(X^{f} - \X^f\bigr)
\\
&&{}+ \frac{1}{4} \bigl(X^{a} - \X^f
\bigr)^T \frac{\partial^2 P} { \partial x^2}(\xi_2) \bigl(X^{a} -
\X^f\bigr)
\\
&=& P\bigl(\X^f\bigr) + \frac{1}{8} \bigl(X^{f} -
X^{a}\bigr)^T \frac{\partial^2 P} { \partial x^2}(\xi_3)
\bigl(X^{f} - X^{a}\bigr)
\end{eqnarray*}
for some $\xi_3$ on the line between $X^{a}$ and $X^{f}$,
due to the mean value theorem.
We then obtain
\begin{eqnarray*}
\frac{1}{2} \bigl(P\bigl(X^{f}\bigr) + P\bigl(X^{a}
\bigr)\bigr) - P\bigl(X^{c}\bigr) &=& \frac{\partial P}{\partial x}^T (
\xi_4) \bigl(\X^f - X^c\bigr)
\\
&&{}  + \frac{1}{8} \bigl(X^{f} - X^{a}\bigr)^T
\frac{\partial^2 P} { \partial x^2}(\xi_3) \bigl(X^{f} - X^{a}
\bigr),
\end{eqnarray*}
%
for some $\xi_4$ on the line between $\X^f$ and $X^{c}$.
Hence,
\[
\bigl\llvert \tfrac{1}{2} \bigl(P\bigl(X^{f}\bigr) + P
\bigl(X^{a}\bigr)\bigr) - P\bigl(X^{c}\bigr) \bigr\rrvert
\leq L_1 \bigl\llVert \X^f - X^c \bigr\rrVert
+ \tfrac{1}{4} L_2 \bigl\llVert X^{f} -
X^{a} \bigr\rrVert ^{2}
\]
and the final result follows from the standard inequality
%
\begin{equation}\label{eqineq}
\Biggl\llvert \sum_{n=1}^N
a_n \Biggr\rrvert ^p \leq N^{p-1}\sum
_{n=1}^N |a_n|^p
\end{equation}
and then taking the expectation.
\end{pf}

In the multi-dimensional SDE applications considered in this paper,
we will show that the Milstein approximation with the L{\'e}vy areas
set to zero, combined with the antithetic construction, leads to
$X^{f}-X^{a} = O(\Delta t^{1/2})$ but $ \X^f-X^{c} = O(\Delta t)$.
Hence, the variance
$\mathbb{V}[\frac{1}{2} (P^{f}_{l} + P^{a}_{l}) -P^c_{l-1}]$
is $O(\Delta t^2)$, which is the order obtained for scalar SDEs using
the Milstein discretisation with its first order strong convergence.
We first show this for the simple Clark and Cameron model problem
which can be analysed in detail. We then extend the analysis to a
general class of multi-dimensional SDEs.

\section{Clark--Cameron example} \label{secCC}

\subsection{Clark--Cameron analysis}

The paper of Clark and Cameron \cite{cc80} addresses the question
of how accurately one can approximate the solution of an SDE driven
by an underlying multi-dimensional Brownian motion, using only
uniformly-spaced discrete Brownian increments. Their model problem is
%
\begin{eqnarray}\label{eqClark}
\mathrm{d}x_1(t) &=& \mathrm{d}w_{1}(t),
\nonumber\\[-8pt]\\[-8pt]
\mathrm{d}x_2(t) &=& x_1(t) \,\mathrm{d}w_2(t)\nonumber
\end{eqnarray}
with $x(0)=y(0)=0$, and zero correlation between the two Brownian
motions $w_1(t)$ and $w_2(t)$. These equations can be integrated exactly
over a time interval $[t_n,t_{n+1}]$, where $t_n = n \Delta t$, to give
%
\begin{eqnarray}
\label{eqCC} x_1(t_{n+1}) &=& x_1(t_n)
+ \Delta w_{1,n},
\nonumber\\[-8pt]\\[-8pt]
x_2(t_{n+1}) &=& x_2(t_n) +
x_1(t_n)\Delta w_{2,n} + \tfrac{1}{2}
\Delta w_{1,n}\Delta w_{2,n} + \tfrac{1}{2}
A_{12,n},\nonumber
\end{eqnarray}
where $\Delta w_{i,n} \equiv w_i(t_{n+1})-w_i(t_n)$, and $A_{12,n}$ is the
L{\'e}vy area defined as
\[
A_{12,n} = \int_{t_n}^{t_{n+1}} \bigl(
w_1(t) - w_1(t_n)
\bigr) \,\mathrm{d} w_2(t) - \int_{t_n}^{t_{n+1}}
\bigl(w_2(t) - w_2(t_n)
\bigr) \,\mathrm{d} w_1(t).
\]
This corresponds exactly to the Milstein discretisation presented
in (\ref{eqMilstein}), so for this simple model problem,
the Milstein discretisation is exact.

The point of Clark and Cameron's paper is that for a given set
of discrete Brownian increments, the value for $x_1(t_n)$ is
determined exactly for all $n$, but the value for $x_2(t_n)$
depends on the unknown L{\'e}vy areas. Since
$
\mathbb{E}[A_{12,n} | \Delta w_{1,n}, \Delta w_{2,n}] = 0$,
the conditional expected value is given by (\ref{eqCC}) with the L{\'e}vy
areas set to zero. In addition, it follows that for \textit{any} numerical
approximation $X(T)$ based solely on the set of discrete Brownian
increments~$\Delta w$,
\begin{eqnarray*}
\mathbb{E}\bigl[\bigl(x_2(T) - X_2(T)
\bigr)^2\bigr] &=&\mathbb{E}\bigl[ \mathbb{E}\bigl[
\bigl(x_2(T) - X_2(T)\bigr)^2 | \Delta w
\bigr] \bigr]
\\
&\geq&\mathbb{E}\bigl[ \mathbb{V}\bigl[x_2(T) | \Delta w \bigr]
\bigr]
\\
& = &\frac{1}{4} \sum_{n=0}^{N-1}
\mathbb{V}[ A_{12,n} ]
\\
& = & \frac{1}{4} T \Delta t.
\end{eqnarray*}
Hence, one cannot achieve better than $O(\Delta t^{1/2})$ strong
convergence, and the mean square error is minimised when the
inequality in the above equation is an equality, which is when
%
\begin{equation}
\label{eqMMSE} X_2(T) = \mathbb{E}\bigl[ x_2(T) | \Delta
w \bigr],
\end{equation}
which is achieved by setting the L{\'e}vy areas set to zero.

\subsection{Antithetic MLMC estimator}

We define a coarse path approximation $X^c$ with timestep $\Delta t$
by neglecting the L{\'e}vy area terms to give
%
\begin{eqnarray} \label{eqCC1}
X^c_{1,n+1} &=& X^c_{1,n} + \Delta
w_{1,n},
\nonumber\\[-8pt]\\[-8pt]
X^c_{2,n+1} &=& X^c_{2,n} +
X^c_{1,n}\Delta w_{2,n} + \tfrac{1}{2}
\Delta w_{1,n} \Delta w_{2,n}.\nonumber
\end{eqnarray}
This is equivalent to replacing the true Brownian path by a piecewise
linear approximation as illustrated in Figure~\ref{figBrownian}.

%
\begin{figure}

\includegraphics{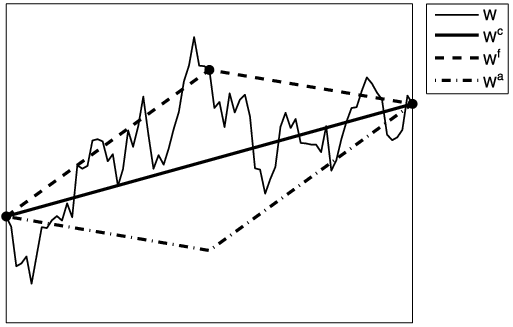}

\caption{Brownian path and approximations over one coarse
timestep.}\label{figBrownian}
\end{figure}

Similarly, we define the corresponding two half-timesteps of
the first fine path approximation $X^{f}$ by
\begin{eqnarray*}
X^{f}_{1,n+\halfs} &=& X^{f}_{1,n} + \delta
w_{1,n},
\\
X^{f}_{2,n+\halfs} &=& X^{f}_{2,n} +
X^{f}_{1,n} \delta w_{2,n} + \tfrac{1}{2}
\delta w_{1,n} \delta w_{2,n},
\\
X^{f}_{1,n+1} &=& X^{f}_{1,n+\halfs} + \delta
w_{1,n+\halfs},
\\
X^{f}_{2,n+1} &=& X^{f}_{2,n+\halfs} +
X^{f}_{1,n+\halfs} \delta w_{2,n+\halfs} + \tfrac{1}{2}
\delta w_{1,n+\halfs} \delta w_{2,n+\halfs}
\end{eqnarray*}
in which
$
\delta w_{n} \equiv w(t_{n+\halfs}) - w(t_{n}), \delta w_{n+\halfs}
\equiv w(t_{n+1}) - w(t_{n+\halfs})
$
are the Brownian increments over the first and second halves of the coarse
timestep, and so $\Delta w_n = \delta w_{n} + \delta w_{n+\halfs}$.
Using this relation, the equations for the two fine timesteps can be
combined to give an equation for the increment over the coarse timestep,
%
\begin{eqnarray}
X^{f}_{1,n+1} &=& X^{f}_{1,n} + \Delta
w_{1,n},\nonumber
\\
X^{f}_{2,n+1} &=& X^{f}_{2,n} +
X^{f}_{1,n} \Delta w_{2,n} + \tfrac{1}{2}
\Delta w_{1,n} \Delta w_{2,n} \label{eqCC2}
\\
&&{} + \tfrac{1}{2} ( \delta w_{1,n} \delta w_{2,n+\halfs} -
\delta w_{2,n} \delta w_{1,n+\halfs} ).\nonumber
\nonumber
\end{eqnarray}

The antithetic approximation $X_n^{a}$ is defined by exactly
the same discretisation except that the Brownian increments
$\delta w_{n}$ and $\delta w_{n+\halfs}$ are swapped,
as illustrated in Figure~\ref{figBrownian}. This gives
\begin{eqnarray*}
X^{a}_{1,n+\halfs} &=& X^{a}_{1,n} + \delta
w_{1,n+\halfs},
\\
X^{a}_{2,n+\halfs} &=& X^{a}_{2,n} +
X^{a}_{1,n} \delta w_{2,n+\halfs} + \tfrac{1}{2}
\delta w_{1,n+\halfs}\delta w_{2,n+\halfs},
\\
X^{a}_{1,n+1} &=& X^{a}_{1,n+\halfs} + \delta
w_{1,n},
\\
X^{a}_{2,n+1} &=& X^{a}_{2,n+\halfs} +
X^{a}_{1,n+\halfs} \delta w_{2,n} + \tfrac{1}{2}
\delta w_{1,n} \delta w_{2,n}
\end{eqnarray*}
and hence
%
\begin{eqnarray}
X^{a}_{1,n+1} &=& X^{a}_{1,n} + \Delta
w_{1,n},\nonumber
\\
X^{a}_{2,n+1} &=& X^{a}_{2,n} +
X^{a}_{1,n} \Delta w_{2,n} + \tfrac{1}{2}
\Delta w_{1,n} \Delta w_{2,n} \label{eqCC3}
\\
&&{} - \tfrac{1}{2} ( \delta w_{1,n} \delta w_{2,n+\halfs} -
\delta w_{2,n} \delta w_{1,n+\halfs} ).\nonumber
\end{eqnarray}
Swapping $\delta w_{n}$ and $\delta w_{n+\halfs}$ does not change the
distribution of the driving Brownian increments, and hence $X^{a}$
has exactly the same distribution as $X^{f}$. Note also the
change in sign in the last term in (\ref{eqCC2}) compared to
the corresponding term in~(\ref{eqCC3}). This is important
because these two terms cancel when the two equations are averaged.

These last terms correspond to the L{\'e}vy areas for the fine path and
the antithetic path, and the sign reversal is a particular instance of
a more general result for time-reversed Brownian motion, \cite{ks91}. If
$(w_{t}, 0\le t \le1 )$ denotes a Brownian motion
on the time interval $[0,1]$, then the time-reversed Brownian
motion $(z_{t}, 0\le t \le1 )$ defined by
%
\begin{equation}
\label{eqTR} z_{t} = w_{1} - w_{1-t},
\end{equation}
has exactly the same distribution, and it can be shown that its
L{\'e}vy area is equal in magnitude and opposite in sign to that of $w_t$.

%
%
\begin{lemma} \label{lemCCunbiased}
If $X_n^{f}$, $X_n^{a}$ and $X_n^c$ are as defined above, then
\[
X_{1,n}^{f} = X_{1,n}^{a} =
X_{1,n}^c,\qquad \tfrac{1}{2} \bigl(
X_{2,n}^{f} + X_{2,n}^{a} \bigr) =
X_{2,n}^c\qquad\forall n \leq N
\]
and
\[
\mathbb{E} \bigl[ \bigl( X_{2,N}^{f} - X_{2,N}^{a}
\bigr)^4 \bigr] = \tfrac{3}{4} T (T + \Delta t) \Delta
t^2.
\]
\end{lemma}

\begin{pf}
Comparing (\ref{eqCC1}), (\ref{eqCC2}) and (\ref{eqCC3}), it
is clear that $X_{1,n}^{f}$, $X_{1,n}^{a}$ and $X_{1,n}^c$ all satisfy
the same difference equation and so are equal.
Given this, averaging the equations for $X_{2,n}^{f}$ and $X_{2,n}^{a}$
gives the same difference equation as for $X_{2,n}^{c}$, and so therefore
$\frac{1}{2}  ( X_{2,n}^{f} + X_{2,n}^{a}  ) = X_{2,n}^c$.
Finally, summing the difference of the equations for $X_{2,n}^{f}$ and
$X_{2,n}^{a}$ gives
\[
X_{2,N}^{f} - X_{2,N}^{a} = \sum
_{n=0}^{N-1} ( \delta w_{1,n}
\delta w_{2,n+\halfs} - \delta w_{2,n} \delta w_{1,n+\halfs} ).
\]
Since the $\delta w_{j,n}$ are all i.i.d. normal variables with variance
$\frac{1}{2}\Delta t$, it is easily shown that
\begin{eqnarray*}
\mathbb{E}\bigl[(\delta w_{1,n} \delta w_{2,n+\halfs} - \delta
w_{2,n} \delta w_{1,n+\halfs})^2\bigr] &=&\tfrac{1}{2} \Delta t^2,
\\
\mathbb{E}\bigl[(\delta w_{1,n} \delta w_{2,n+\halfs} - \delta
w_{2,n} \delta w_{1,n+\halfs})^4\bigr] &=&
\tfrac{3}{2} \Delta t^4
\end{eqnarray*}
and it then follows that
\[
\mathbb{E}\bigl[ \bigl(X_{2,N}^{f} - X_{2,N}^{a}
\bigr)^4 \bigr] = \biggl( \frac{1}{2} \Delta t^2
\biggr)^2 \frac{N(N - 1)}{2} \frac{4\times3}{2} + \frac{3}{2}
\Delta t^4 N = \frac{3}{4} T (T + \Delta t) \Delta
t^2.
\]
In the above derivation, when expanding $(X_{2,N}^{f} - X_{2,N}^{a})^4$,
the first contribution comes from terms of the form
$
(\delta w_{1,m} \delta w_{2,m+\halfs} - \delta w_{2,m} \delta
w_{1,m+\halfs})^2\*
(\delta w_{1,n} \delta w_{2,n+\halfs} - \delta w_{2,n} \delta
w_{1,n+\halfs})^2
$
for $m\neq n$, while the second contribution comes from terms of the form
$
(\delta w_{1,n} \delta w_{2,n+\halfs} - \delta w_{2,n} \delta
w_{1,n+\halfs})^4
$.
All other terms have zero expectation.
\end{pf}

Combining the above result with Lemma \ref{lemmapayoff} for
$p = 2$ gives a second order bound on the multilevel estimator
variance for payoffs satisfying the required smoothness conditions.
It is worth noting that an antithetic MLMC based on the simpler
Euler--Maruyama discretisation, omitting the term $\Delta w_{1,n}
\Delta w_{2,n}$
in~(\ref{eqCC1}), would not give similar benefits. The identity
$\frac{1}{2}  ( X_{2,n}^{f} + X_{2,n}^{a}  ) = X_{2,n}^c$
no longer holds, and a similar analysis to that in the proof above gives
\[
\mathbb{V} \bigl[ \tfrac{1}{2} \bigl( X_{2,N}^{f} +
X_{2,N}^{a} \bigr) - X_{2,N}^c \bigr]
= \mathbb{E} \bigl[ \bigl( \tfrac{1}{2} \bigl( X_{2,N}^{f}
+ X_{2,N}^{a} \bigr) - X_{2,N}^c
\bigr)^2 \bigr] = O(\Delta t).
\]
Hence, in the simple case in which the payoff is $P=X_2(T)$, the variance
of the antithetic multilevel estimator is first order, the same as the
standard MLMC, and not second order.

\section{General theory} \label{secGeneral}

\subsection{Milstein discretisation}

In this section we extend the analysis of the Clark--Cameron example
to general the multi-dimensional SDE (\ref{eqSDE}).
We make the standard assumptions that $f$, $g$ and $h$ have
a uniform Lipschitz bound, and so have uniformly bounded first
derivatives. In addition, we make the assumption that\vadjust{\goodbreak}
$f$ and $g$ have uniformly bounded second derivatives.
More formally, we have the following:

%
\begin{assp} \label{asgLip}
Let $f\in C^2 (\mathbb{R}^{d}, \mathbb{R}^{d})$ and $g\in C^2
(\mathbb{R}^{d}, \mathbb{R}^{d\times D})$.
There exists a constant $L$ such that for any $x\in\mathbb{R}^{d}$,
and for all $1\leq i \leq d$ and $1\leq j, k, l \leq D$,
\begin{eqnarray*}
\biggl\llvert \frac{\partial f_i} {\partial x_l}(x) \biggr\rrvert &\leq& L,\qquad \biggl\llvert
\frac{\partial g_{ij}} {\partial x_l}(x) \biggr\rrvert \leq L,\qquad \biggl\llvert
\frac{\partial h_{ijk}}{\partial x_l}(x) \biggr\rrvert \leq L,
\\
\biggl\llvert\frac{\partial^2 f_i} {\partial x_k
\partial x_l}(x) \biggr\rrvert &\leq& L,\qquad \biggl\llvert
\frac{\partial^2 g_{ij}}{\partial x_k\,\partial x_l}(x) \biggr\rrvert \leq L.
\end{eqnarray*}
\end{assp}

Let us recall that the general Milstein scheme \cite{kp92} has the form
%
\begin{eqnarray}\label{eqgenMilst}
\widehat{X}_{i,n+1} &=& \widehat{X}_{i,n} +
f_i(\widehat{X}_n) \Delta t + \sum
_{j=1}^D g_{ij}(\widehat{X}_n)
\Delta w_{j,n}
\nonumber\\[-8pt]\\[-8pt]
&&{} + \sum_{j,k=1}^D h_{ijk}(\widehat{X}_n) (\Delta w_{j,n} \Delta
w_{k,n} - \Omega_{jk} \Delta t - A_{jk,n} ).\nonumber
\end{eqnarray}
As in the Clark--Cameron example, we drop the L{\'e}vy areas terms,
and instead use the truncated Milstein approximation
%
\begin{eqnarray}\label{eqreducedMilstein}
X_{i,n+1} &=& X_{i,n} + f_i(X_n)
\Delta t + \sum_{j=1}^D
g_{ij}(X_n) \Delta w_{j,n}
\nonumber\\[-8pt]\\[-8pt]
&&{} + \sum_{j,k=1}^D h_{ijk}(X_n) (\Delta
w_{j,n} \Delta w_{k,n} - \Omega_{jk} \Delta t ).\nonumber
\end{eqnarray}

Under Assumption \ref{asgLip} it is a standard result that the
moments of the general Milstein approximation $\widehat{X}_n$ are bounded,
and $\widehat{X}_n$ strongly converges to the solution of the SDE (\ref{eqSDE});
this remains true for the truncated Milstein approximation as
stated in the following lemma.
%
%
\begin{lemma} \label{lemmastrong}
For $p\ge2$, there
exists a constant $K_p$, independent of the time step, such that
\[
\mathbb{E} \Bigl[ \max_{0\leq n\leq N} \| X_n
\|^p \Bigr] \leq K_p
\]
and
\[
\mathbb{E} \Bigl[ \max_{0\leq n\leq N} \bigl\| X_n -
x(t_n) \bigr\|^p \Bigr] \leq K_p \Delta
t^{p/2}.
\]
\end{lemma}
\begin{pf}
The proof in \cite{muller02} follows the standard
method of analysis in references such as \cite{kp92,mt04}.\vadjust{\goodbreak}
\end{pf}

Hence, the rate of strong convergence is $O(\Delta t^{1/2})$, which is
no better than the Euler--Maruyama discretisation. Nevertheless,
we will show that the antithetic multilevel estimator has a variance which
converges to zero at the same rate as the full Milstein approximation.

%
\begin{coro} \label{corostrong}
For $p\ge2$, there
exists a constant $K_p$, independent of the time step, such that
\begin{eqnarray*}
\mathbb{E} \Bigl[ \max_{0\leq n\leq N} \bigl| f_i(X_n)
\bigr|^p \Bigr] &\leq& K_p,\qquad \mathbb{E} \Bigl[ \max
_{0\leq n\leq N} \bigl| g_{ij}(X_n) \bigr|^p
\Bigr] \leq K_p,
\\
\mathbb{E} \Bigl[ \max_{0\leq n\leq N}
\bigl| h_{ijk}(X_n) \bigr|^p \Bigr] &\leq& K_p
\end{eqnarray*}
for all $1\leq i \leq d$ and $1 \leq j, k \leq D$.
\end{coro}
\begin{pf}
The bounded first derivatives of $f(x), g(x), h(x)$ imply that they
grow no faster than linearly as $\|x\| \rightarrow\infty$, and the
result then follows from the bound in Lemma \ref{lemmastrong}.
\end{pf}

In order to derive appropriate bounds on the antithetic estimator
we also need the following lemma.
%
%
\begin{lemma} \label{lemonestep}
For $p\ge2$, there exists a constant
$K_p$, independent of the time step, such that
\[
\max_{0\leq n\leq N} \mathbb{E} \bigl[\llVert
X_{n+1} - X_n \rrVert ^p \bigr] \leq
K_p \Delta t^{p/2}.
\]
\end{lemma}

\begin{pf}
We
start from (\ref{eqreducedMilstein}) and inequality (\ref{eqineq})
which gives
\begin{eqnarray*}
\mathbb{E} \bigl[ \llvert X_{i,n+1} - X_{i,n} \rrvert
^{p} \bigr] &\le& 3^{p-1} \Biggl(\mathbb{E} \bigl[ \bigl
\llvert f_i(X_n) \Delta t \bigr\rrvert ^{p}
\bigr]+ \mathbb{E} \Biggl[ \Biggl\llvert \sum_{j=1}^D
g_{ij}(X_n) \Delta w_{j,n} \Biggr\rrvert
^{p} \Biggr]
\\
&&\hspace*{26pt}{} + \mathbb{E} \Biggl[ \Biggl\llvert \sum
_{j,k=1}^D h_{ijk}(X_n) (\Delta
w_{j,n}\Delta w_{k,n} - \Omega_{jk} \Delta t ) \Biggr
\rrvert ^{p} \Biggr] \Biggr).
\end{eqnarray*}
The first term on the right has a $O(\Delta t^p)$ bound due to
the uniform bound on $\mathbb{E} [ \llvert  f_i(X_n) \rrvert ^{p} ]$.
For the second term we note that because $\Delta w_{j,n}$ is
independent of~$X_n$, then
\[
\mathbb{E} \Biggl[ \Biggl\llvert \sum_{j=1}^D
g_{ij}(X_n) \Delta w_{j,n} \Biggr\rrvert
^{p} \Biggr] \leq D^{p-1} \sum_{j=1}^D
\mathbb{E} \bigl[ \bigl\llvert g_{ij}(X_n)\bigr\rrvert
^{p} \bigr] \mathbb{E} \bigl[ \llvert \Delta w_{j,n} \rrvert
^{p} \bigr]
\]
and we obtain a $O(\Delta t^{p/2})$ bound due to the uniform bound
on $\mathbb{E} [ \llvert  g_{ij}(X_n) \rrvert ^{p} ]$ and
standard results for
the moments of Brownian increments. The third term is handled in a
similar way and has a $O(\Delta t^p)$ bound.

Together these give a $O(\Delta t^{p/2})$ bound for
$\mathbb{E} [ \llvert  X_{i,n+1} - X_{i,n} \rrvert ^{p}
]$ for each $i$,
and hence also for
$\mathbb{E} [ \llVert  X_{n+1} - X_{n} \rrVert ^{p}  ]$.
\end{pf}
%
\subsection{Antithetic MLMC estimator}
Using the coarse timestep $\Delta t$, the coarse path approximation $X_n^c$,
is given by the Milstein approximation without the L{\'e}vy area term,
\begin{eqnarray*}
X^c_{i,n+1} &=& X^c_{i,n} +
f_i\bigl(X^c_n\bigr) \Delta t + \sum
_{j=1}^D g_{ij}
\bigl(X^c_n\bigr) \Delta w_{j,n}
\\[-2pt]
&&{}  + \sum_{j,k=1}^D h_{ijk}\bigl(X^c_n
\bigr) (\Delta w_{j,n} \Delta w_{k,n} - \Omega_{jk}
\Delta t ).
\end{eqnarray*}
The first fine path approximation $X_n^{f}$ uses the corresponding
discretisation with timestep $\Delta t/2$,
%
\begin{eqnarray}
\qquad X^{f}_{i,n+\halfs} &=& X^{f}_{i,n}
+ f_i\bigl(X^{f}_n\bigr) \Delta t/2 + \sum
_{j=1}^D g_{ij}
\bigl(X^{f}_n\bigr) \delta w_{j,n}
\nonumber\\[-10pt]\label{eqeqn1} \\[-10pt]
&&{} + \sum_{j,k=1}^D h_{ijk}
\bigl(X^{f}_n\bigr) (\delta w_{j,n} \delta
w_{k,n} - \Omega_{jk} \Delta t/2 ),\nonumber
\\
X^{f}_{i,n+1} &=& X^{f}_{i,n+\halfs} +
f_i\bigl(X^{f}_{n+\halfs}\bigr) \Delta t/2 + \sum
_{j=1}^D g_{ij}
\bigl(X^{f}_{n+\halfs}\bigr) \delta w_{n+\halfs}
\nonumber\\[-10pt] \label{eqfine+}\\[-10pt]
&&{} + \sum_{j,k=1}^D h_{ijk}
\bigl(X^{f}_{n+\halfs}\bigr) (\delta w_{j,n+\halfs} \delta
w_{k,n+\halfs} - \Omega_{jk} \Delta t/2 ),\nonumber
\end{eqnarray}
in which
%
\begin{equation}
\label{eqincrements} \delta w_{n} \equiv w(t_{n+\halfs}) -
w(t_{n}),\qquad \delta w_{n+\halfs} \equiv w(t_{n+1}) -
w(t_{n+\halfs})
\end{equation}
are the Brownian increments over the first and second halves of the coarse
timestep, and so $\Delta w_n = \delta w_{n} + \delta w_{n+\halfs}$.

The antithetic approximation $X_n^{a}$ is defined by exactly the same
discretisation, except that the Brownian increments $\delta w_{n}$ and
$\delta w_{n+\halfs}$ are swapped, so that
%
\begin{eqnarray}
\label{eqeqn2} X^{a}_{i,n+\halfs} &=& X^{a}_{i,n}
+ f_i\bigl(X^{a}_n\bigr) \Delta t/2 + \sum
_{j=1}^D g_{ij}
\bigl(X^{a}_n\bigr) \delta w_{n+\halfs}
\nonumber
\\[-2pt]
&&{} + \sum_{j,k=1}^D h_{ijk}
\bigl(X^{a}_n\bigr) (\delta w_{j,n+\halfs} \delta
w_{k,n+\halfs} - \Omega_{jk} \Delta t/2 ),
\nonumber\\[-10pt]\label{eqfine-} \\[-10pt]
X^{a}_{i,n+1} &=& X^{a}_{i,n+\halfs} +
f_i\bigl(X^{a}_{n+\halfs}\bigr) \Delta t/2 + \sum
_{j=1}^D g_{ij}
\bigl(X^{a}_{n+\halfs}\bigr) \delta w_{j,n}\nonumber
\\[-2pt]
&&{} + \sum_{j,k=1}^D h_{ijk}
\bigl(X^{a}_{n+\halfs}\bigr) (\delta w_{j,n} \delta
w_{k,n} - \Omega_{jk} \Delta t/2 ).
\nonumber
\end{eqnarray}

Since $\delta w_{n}$ and $\delta w_{n+\halfs}$ are independent
and identically distributed, $X^{a}$ has exactly the same
distribution as $X^{f}$, and hence $\mathbb{E}[P(X^{a})] = \mathbb
{E}[P(X^{f})]$.
In addition, the following lemma follows directly from
Lemmas \ref{lemmastrong} and \ref{lemonestep}.

%
\begin{lemma} \label{lemonestep2}
Let $X^{f}$ and $X^{a}$ be as defined above.
Then for $p\ge2$, there exists a constant $K_p$,
independent of the time step, such that
\begin{eqnarray*}
\mathbb{E} \Bigl[\max_{0\leq n\leq N} \bigl\llVert
X^{f}_n \bigr\rrVert ^p \Bigr] &\leq&
K_p,\qquad  \max_{0\leq n< N} \mathbb{E} \bigl[\bigl\llVert X^{f}_{n+\halfs} - X^{f}_n
\bigr\rrVert ^p \bigr] \leq K_p \Delta t^{p/2},
\\
\mathbb{E} \Bigl[\max_{0\leq n\leq N} \bigl\llVert
X^{a}_n \bigr\rrVert ^p \Bigr] &\leq&
K_p,\qquad  \max_{0\leq n< N} \mathbb{E} \bigl[\bigl\llVert X^{a}_{n+\halfs} - X^{a}_n
\bigr\rrVert ^p \bigr] \leq K_p \Delta t^{p/2}.
\end{eqnarray*}
\end{lemma}

\subsection{Numerical analysis}

The analysis is presented as a sequence of lemmas and theorems, with
the proofs deferred to the \hyperref[app]{Appendix}. The outline is as follows:
\begin{itemize}
\item
Lemma \ref{lemmadiff} bounds
$\| X_n^{f} - X_n^{a}\|$ over a coarse timestep;

\item
Lemma \ref{lemmaeqn1} gives a representation of the discrete equation
for $X_n^{f}$ over a coarse timestep, and Corollary \ref{coroeqn1}
gives the corresponding representation for $X_n^{a}$;

\item
Lemma \ref{lemmaeqn2} gives a representation
of the discrete equation describing the evolution of the average
$\X_n^f = \frac{1}{2}(X_n^{f} + X_n^{a})$ over a coarse timestep;

\item
Theorem \ref{thmmain} bounds $\|\X_n^f - X_n^{c}\|$ over a
coarse timestep.
\end{itemize}

%
\begin{lemma} \label{lemmadiff}
For all integers $p\geq2$, there exists a constant $K_p$ such that
\[
\mathbb{E} \Bigl[ \max_{0\leq n\leq N} \bigl\| X^{f}_n
- X^{a}_n \bigr\|^p \Bigr] \leq K_p
\Delta t^{p/2}.
\]
\end{lemma}

%
\begin{lemma} \label{lemmaeqn1}
Difference equation (\ref{eqfine+}) for $X_n^{f}$ can be expressed as
\begin{eqnarray*}
X^{f}_{i,n+1} &=& X^{f}_{i,n} +
f_i\bigl(X^{f}_n\bigr) \Delta t + \sum
_{j=1}^D g_{ij}
\bigl(X^{f}_n\bigr) \Delta w_{j,n}
\\
&&{} + \sum
_{j,k=1}^D h_{ijk}
\bigl(X^{f}_n\bigr) (\Delta w_{j,n}\Delta
w_{k,n} - \Omega_{jk} \Delta t )
\\
&&{} - \sum_{j,k=1}^D h_{ijk}
\bigl(X^{f}_n\bigr) ( \delta w_{j,n} \delta
w_{k,n+\halfs} - \delta w_{k,n} \delta w_{j,n+\halfs} )
\\
&&{} + M^{f}_{i,n} + N^{f}_{i,n},
\end{eqnarray*}
where 
$\mathbb{E}[ M^{f}_n | \mathcal{F}_n ] = 0$,
and for any integer $p \geq 2$ there exists a constant $K_p$
such that
\[
\max_{0\leq n\leq N} \mathbb{E} \bigl[\bigl\|
M^{f}_{n} \bigr\|^p \bigr] \leq K_p
\Delta t^{3p/2},\qquad \max_{0\leq n\leq N} \mathbb{E} \bigl[
\bigl\| N^{f}_{n} \bigr\|^p \bigr]
\leq K_p \Delta t^{2p}.
\]
\end{lemma}

%
\begin{coro} \label{coroeqn1}
Difference equation (\ref{eqfine-}) for $X_n^{a}$ can be expressed as
\begin{eqnarray*}
X^{a}_{i,n+1} &=& X^{a}_{i,n} +
f_i\bigl(X^{a}_n\bigr) \Delta t + \sum
_{j=1}^D g_{ij}
\bigl(X^{a}_n\bigr) \Delta w_{j,n}
\\
&&{} + \sum
_{j,k=1}^D h_{ijk}
\bigl(X^{a}_n\bigr) (\Delta w_{j,n}\Delta
w_{k,n} - \Omega_{jk} \Delta t )
\\
&&{} + \sum_{j,k=1}^D h_{ijk}
\bigl(X^{a}_n\bigr) ( \delta w_{j,n} \delta
w_{k,n+\halfs} - \delta w_{k,n} \delta w_{j,n+\halfs} )
\\
&&{} + M^{a}_{i,n} + N^{a}_{i,n},
\end{eqnarray*}
where
$\mathbb{E}[ M^{a}_n | \mathcal{F}_n ] = 0$,
and for any integer $p \geq 2$ there exists a constant $K_p$
such that
\[
\max_{0\leq n\leq N} \mathbb{E} \bigl[ \bigl\|
M^{a}_{n} \bigr\|^p \bigr] \leq K_p
\Delta t^{3p/2},\qquad \max_{0\leq n\leq N} \mathbb{E} \bigl[
\bigl\| N^{a}_{n} \bigr\|^p \bigr]
\leq K_p \Delta t^{2p}.
\]
\end{coro}

%
\begin{lemma} \label{lemmaeqn2}
The difference equation for
$\X_n^f\equiv\frac{1}{2}(X_n^{f} + X_n^{a})$
can be expressed~as
\begin{eqnarray*}
\X_{i,n+1}^f &=& \X^{f}_{i,n} +
f_i\bigl(\X^{f}_n\bigr) \Delta t + \sum
_{j=1}^D g_{ij}\bigl(
\X^{f}_n\bigr) \Delta w_{j,n}
\\
&&{} + \sum
_{j,k=1}^D h_{ijk}\bigl(\X^{f}_n
\bigr) (\Delta w_{j,n} \Delta w_{k,n} - \Omega_{jk}
\Delta t )
\\
&&{} + M_{i,n} + N_{i,n},
\end{eqnarray*}
where 
$\mathbb{E}[ M_n | \mathcal{F}_n ] = 0$,
and for any integer $p \geq 2$ there exists a constant $K_p$
such that
\[
\max_{0\leq n\leq N} \mathbb{E} \bigl[\|
M_{n} \|^p \bigr] \leq K_p \Delta
t^{3p/2},\qquad \max_{0\leq n\leq N} \mathbb{E} \bigl[\| N_{n} \|^p \bigr] \leq K_p \Delta
t^{2p}.
\]
\end{lemma}

%
\begin{theorem} \label{thmmain}
For all $p\geq2$, there exists a constant $K_p$ such that
\[
\mathbb{E} \Bigl[ \max_{0\leq n\leq N} \bigl\| \X_n^f
- X_n^c \bigr\|^p \Bigr] \leq K_p
\Delta t^p.
\]
\end{theorem}

\subsection{Piecewise linear interpolation analysis}

The piecewise linear interpolant $X^{c}(t)$ for the coarse path
is defined within the coarse timestep interval $[t_{k}, t_{k+1}]$ as
\[
X^{c}(t) \equiv (1 - \lambda) X^{c}_k +
\lambda X^{c}_{k+1}, \qquad \lambda\equiv\frac{t - t_k}{ t_{k+1}-t_{k}}.
\]
Likewise, the piecewise linear interpolants
$X^{f}(t)$ and $X^{a}(t)$
are defined on the fine timestep $[t_{k}, t_{k+\halfs}]$ as
\begin{eqnarray*}
X^{f}(t) &\equiv& (1 - \lambda) X^{f}_{k} +
\lambda X^{f}_{k+\halfs}, \qquad X^{a}(t) \equiv (1 -
\lambda) X^{a}_{k} + \lambda X^{a}_{k+\halfs},
\\
\lambda &\equiv& \frac{t - t_k}{t_{k+\halfs}-t_{k}}
\end{eqnarray*}
and there is a corresponding definition for the fine timestep
$[t_{k+\halfs}, t_{k+1}]$.

The proofs of the next two lemmas are in the \hyperref[app]{Appendix}, and the
theorem then follows directly.

%
\begin{lemma} \label{corodiff}
For all integers $p\geq2$, there exists a constant $K_p$ such that
\[
\max_{0\leq n<N} \mathbb{E} \bigl[ \bigl\| X^{f}_{n+\halfs}
- X^{a}_{n+\halfs} \bigr\|^p \bigr] \leq K_p
\Delta t^{p/2}.
\]
\end{lemma}

%
\begin{lemma} \label{coromain}
For all $p\geq2$, there exists a constant $K_p$ such that
\[
\max_{0\leq n<N} \mathbb{E} \bigl[ \bigl\llVert
\X^f_{n+\halfs} - X^c(t_{n+\halfs}) \bigr\rrVert
^p \bigr] \leq K_p \Delta t^p,
\]
where $X^c(t_{n+\halfs}) = \frac{1}{2}(X^c_n+X^c_{n+1})$ is
the midpoint value of the coarse path interpolant.
\end{lemma}

%
\begin{theorem} \label{thminterp}
For all $p\geq2$, there exists a constant $K_p$ such that
\begin{eqnarray*}
\sup_{0\leq t\leq T} \mathbb{E} \bigl[ \bigl\| X^{f}(t) -
X^{a}(t) \bigr\|^p \bigr] &\leq& K_p \Delta
t^{p/2},
\\
\sup_{0\leq t\leq T} \mathbb{E} \bigl[ \bigl\llVert \X^f(t)
- X^c(t) \bigr\rrVert ^p \bigr] &\leq& K_p
\Delta t^p,
\end{eqnarray*}
where $\X^f(t)$ is the average of the piecewise linear
interpolants $X^f(t)$ and $X^a(t)$.
\end{theorem}
%

\section{European and Asian payoffs}

\subsection{European options}

In the case of payoff which is a smooth function of the final
state $x(T)$, taking
$p = 2$ in Lemma \ref{lemmapayoff},
$p = 4$ in Lemma \ref{lemmadiff} and
$p = 2$ in Theorem \ref{thmmain},
immediately gives the result that the multilevel variance
\[
\mathbb{V} \bigl[ \tfrac{1}{2} \bigl( P\bigl(X_N^{f}
\bigr) + P\bigl(X_N^{a}\bigr) \bigr) - P
\bigl(X_N^c\bigr) \bigr]
\]
has an $O(\Delta t^2)$ upper bound. This matches the convergence rate
for the multilevel method for scalar SDEs using the standard first order
Milstein discretisation,
and is much better than the $O(\Delta t)$ convergence obtained with the
Euler--Maruyama discretisation.

However, very few financial payoff functions are twice differentiable
on the entire domain $\mathbb{R}^d$. A more typical 2D example is a call
option based on the minimum of two assets,
\[
P\bigl(x(T)\bigr) \equiv\max \bigl(0, \min\bigl( x_1(T),
x_2(T)\bigr) - K \bigr),
\]
which is piecewise linear, with a discontinuity in the gradient along
the three lines
$(s, K)$, $(K, s)$ and $(s, s)$ for $s \geq K$.

To handle such payoffs, we introduce a new assumption which bounds the
probability of the solution of the SDE having a value at time $T$ close
to such lines with discontinuous gradients, and then formulate a theorem
to show that the multilevel variance which results from using the
antithetic estimator has an upper bound which is almost $O(\Delta t^{3/2})$.

%
\begin{assp} \label{aspLip}
The payoff function $P \in C(\mathbb{R}^d, \mathbb{R})$ has a uniform
Lipschitz bound,
so that there exists a constant $L$ such that
\[
\bigl\llvert P(x) - P(y) \bigr\rrvert \leq L \llvert x - y \rrvert \qquad
\forall x, y \in\mathbb{R}^d
\]
and the first and second derivatives exist, are continuous and have
uniform bound $L$ at all points $x \notin K$, where $K$ is a set of
zero measure, and there exists a constant $c$ such that the probability
of the SDE solution $x(T)$, being within a neighbourhood of the set $K$,
has the bound
\[
\mathbb{P} \Bigl( \min_{y\in K} \bigl\| x(T) - y \bigr\| \leq\varepsilon
\Bigr) \leq c \varepsilon \qquad\forall \varepsilon> 0.
\]
\end{assp}
In a 1D context, Assumption \ref{aspLip} corresponds to an assumption
of a locally bounded density for $x(T)$.

%
\begin{theorem} \label{thmLip}
If the SDE satisfies the conditions of Assumption \ref{asgLip},
and the payoff satisfies Assumption \ref{aspLip}, then
\[
\mathbb{E} \bigl[ \bigl(\tfrac{1}{2} \bigl(P\bigl(X_N^{f}
\bigr) + P\bigl(X_N^{a}\bigr)\bigr) - P
\bigl(X_N^c\bigr) \bigr)^2 \bigr] = o\bigl(
\Delta t^{3/2-\delta}\bigr)
\]
for any $\delta>0$.
\end{theorem}

\begin{pf}
We start by noting that
\begin{eqnarray*}
&& \mathbb{E} \bigl[ \bigl( \tfrac{1}{2} \bigl(P\bigl(X_N^{f}
\bigr) + P\bigl(X_N^{a}\bigr)\bigr) - P
\bigl(X_N^c\bigr) \bigr)^{2} \bigr]
\\
&&\qquad  \leq
2 \mathbb{E} \bigl[ \bigl( \tfrac{1}{2} \bigl(P\bigl(X_N^{f}
\bigr) + P\bigl(X_N^{a}\bigr)\bigr) - P\bigl(
\X_N^f\bigr) \bigr)^{2} \bigr]
+ 2 \mathbb{E} \bigl[ \tfrac{1}{2} \bigl(P\bigl(\X_N^f
\bigr) - P\bigl(X_N^c\bigr) \bigr)^{2}
\bigr].
\end{eqnarray*}
The second term on the right-hand side has an $O(\Delta t^2)$ bound
due to the uniform Lipschitz bound for the payoff, together with
the result from Theorem \ref{thmmain} for $p=2$.

The objective now is to prove that the first term has a
$o(\Delta t^{3/2-\delta})$ bound for any $\delta>0$.
The analysis follows the approach used in \cite{ghm09}.
To prove this for a particular value of $\delta$, we define
$
\varepsilon= \Delta t^{1/2-\delta/2}$,
and consider the three events
\begin{eqnarray*}
A &\equiv& \Bigl\{\min_{y\in K}\bigl\| x(T) - y \bigr\| \leq\varepsilon
\Bigr\},
\\
B &\equiv& \bigl\{\bigl\| x(T) - X^{f}_N\bigr\| \geq
\tfrac{1}{2} \varepsilon \bigr\},
\\
C &\equiv& \bigl\{ \bigl\| X^{f}_N - X^{a}_N
\bigr\| \geq\tfrac{1}{2} \varepsilon \bigr\}.
\end{eqnarray*}

Using
$\mathbf{1}_{A}$ to indicate the indicator function for event $A$, and
$A^c$ to denote the complement of $A$, we have
\begin{eqnarray*}
&& \mathbb{E} \bigl[ \bigl( \tfrac{1}{2} \bigl(P\bigl(X_N^{f}
\bigr) + P\bigl(X_N^{a}\bigr)\bigr) - P\bigl(
\X_N^f\bigr) \bigr)^{2} \bigr]
\\
&&\qquad =
\mathbb{E} \bigl[ \bigl( \tfrac{1}{2} \bigl(P\bigl(X_N^{f}
\bigr) + P\bigl(X_N^{a}\bigr)\bigr) - P\bigl(
\X_N^f\bigr) \bigr)^{2} \mathbf{1}_{A \cup B \cup C}
\bigr]
\\
&&\quad\qquad{}+ \mathbb{E} \bigl[ \bigl( \tfrac{1}{2} \bigl(P\bigl(X_N^{f}
\bigr) + P\bigl(X_N^{a}\bigr)\bigr) - P\bigl(
\X_N^f\bigr) \bigr)^{2} \mathbf{1}_{A^c \cap B^c \cap C^c}\bigr].
\end{eqnarray*}

Looking at the first of the two terms on the right-hand side,
then H{\"o}lder's inequality gives
\begin{eqnarray*}
&& \mathbb{E} \bigl[ \bigl( \tfrac{1}{2} \bigl(P\bigl(X_N^{f}
\bigr) + P\bigl(X_N^{a}\bigr)\bigr) - P\bigl(
\X_N^f\bigr) \bigr)^{2} \mathbf{1}_{A \cup B \cup C}\bigr]
\\
&&\qquad\leq\mathbb{E} \bigl[ \bigl( \tfrac{1}{2} \bigl(P
\bigl(X_N^{f}\bigr) + P\bigl(X_N^{a}
\bigr)\bigr) - P\bigl(\X_N^f\bigr) \bigr)^{2p}
\bigr]^{1/p} \bigl( \mathbb{P}(A) + \mathbb{P}(B) + \mathbb{P}(C)
\bigr)^{1/q}
\end{eqnarray*}
for any $p,q\geq1$, with $p^{-1}+q^{-1}=1$.
The Markov inequality gives
\[
\mathbb{P}(B) \leq \mathbb{E} \bigl[ \bigl\llVert x(T) - X^{f}_N
\bigr\rrVert ^m \bigr] / \bigl(\tfrac
{1}{2}\varepsilon
\bigr)^m 
\]
for any $m\geq1$. Using the strong convergence property from
Lemma \ref{lemmastrong}, and the definition of $\varepsilon$, we can
take $m$ to be sufficiently large so that
\[
\frac{1}{2} m - \frac{1-\delta}{2} m > \frac{1-\delta}{2}
\]
and hence there exists a constant $c_1$ such that
$
\mathbb{P}(B) \leq c_1 \varepsilon$.
Using Lemma \ref{lemmadiff}, one can obtain a similar bound
$
\mathbb{P}(C) \leq c_2 \varepsilon$,
and then $q$ can be chosen sufficiently
close to 1 so that
\[
\bigl( \mathbb{P}(A) + \mathbb{P}(B) + \mathbb{P}(C) \bigr)^{1/q}
\leq(1 + c_1 + c_2)^{1/q} \Delta
t^{(1/2-\delta/2)/q} = o\bigl(\Delta t^{1/2-\delta}\bigr).
\]
Since
\[
\tfrac{1}{2} \bigl(P\bigl(X_N^{f}\bigr) + P
\bigl(X_N^{a}\bigr)\bigr) - P\bigl(\X_N^f
\bigr) = \tfrac{1}{2} \bigl(P\bigl(X_N^{f}\bigr) - P
\bigl(\X_N^f\bigr) \bigr) + \tfrac{1}{2} \bigl(P
\bigl(X_N^{a}\bigr) - P\bigl(\X_N^f
\bigr) \bigr),
\]
the uniform Lipschitz bound gives
\[
\mathbb{E} \bigl[ \bigl( \tfrac{1}{2} \bigl(P\bigl(X_N^{f}
\bigr) + P\bigl(X_N^{a}\bigr)\bigr) - P\bigl(
\X_N^f\bigr) \bigr)^{2p} \bigr]^{1/p}
\leq L^2 \mathbb{E} \bigl[ \bigl\llVert X_N^{f}
- X_N^{a} \bigr\rrVert ^{2p}
\bigr]^{1/p} \leq c_3 \Delta t
\]
for some constant $c_3$ due to Lemma \ref{lemmadiff}, and hence
\[
\mathbb{E} \bigl[ \bigl( \tfrac{1}{2} \bigl(P\bigl(X_N^{f}
\bigr) + P\bigl(X_N^{a}\bigr)\bigr) - P\bigl(
\X_N^f\bigr) \bigr)^{2} \mathbf{1}_{A \cup B \cup C}
\bigr] = o\bigl( \Delta t^{3/2 - \delta}\bigr).
\]

Lastly, we consider the second term
\[
\mathbb{E} \bigl[ \bigl( \tfrac{1}{2} \bigl(P\bigl(X_N^{f}
\bigr) + P\bigl(X_N^{a}\bigr)\bigr) - P\bigl(
\X_N^f\bigr) \bigr)^{2} \mathbf{1}_{A^c \cap B^c \cap C^c}
\bigr].
\]
Given\vspace*{-1pt} a path sample $\omega\in(B^c \cap C^c)$, if the straight
line between $X_N^{f}$ and $X_N^{a}$ contains a point $y\in K$,
then $\| y - X_N^{f} \|$ and $\| x(T) - X_N^{f} \|$ are both
less than $\varepsilon/2$, and hence $\| x(T) - y \| < \varepsilon$.

Thus, for a path sample $\omega\in(A^c \cap B^c \cap C^c)$, the
straight line between $X_N^{f}$ and $X_N^{a}$ does not contain
any points in $K$. It is therefore possible to perform a second order
truncated Taylor expansion as in the proof of Lemma \ref{lemmapayoff},
and deduce that there exists a constant $c_4$ such that
\[
\mathbb{E} \bigl[ \bigl( \tfrac{1}{2} \bigl(P\bigl(X_N^{f}
\bigr) + P\bigl(X_N^{a}\bigr)\bigr) - P\bigl(
\X_N^f\bigr) \bigr)^{2} \mathbf{1}_{A^c \cap B^c \cap C^c}
\bigr] \leq c_4 \mathbb{E} \bigl[ \bigl\llVert X_N^{f}
- X_N^{a} \bigr\rrVert ^4 \bigr],
\]
which has an $O(\Delta t^2)$ bound due to Lemma \ref{lemmadiff}.
\end{pf}

\subsection{Asian payoffs}

For an Asian option, the payoff depends on the average
\[
x_{\mathrm{ave}} \equiv T^{-1} \int_0^T
x(t) \,\mathrm{d}t.
\]
This can be approximated by integrating the appropriate piecewise linear
interpolant which gives
\begin{eqnarray*}
X^c_{\mathrm{ave}} &\equiv& T^{-1} \int
_0^T X^c(t) \,\mathrm{d}t =
N^{-1} \sum_{n=0}^{N-1}
\frac{1}{2} \bigl(X^c_{n} + X^c_{n+1}
\bigr),
\\
X^{f}_{\mathrm{ave}} &\equiv& T^{-1} \int
_0^T X^{f}(t) \,\mathrm{d}t =
N^{-1} \sum_{n=0}^{N-1}
\frac{1}{4} \bigl(X^{f}_{n} + 2 X^{f}_{n+\halfs}
+ X^{f}_{n+1}\bigr),
\\
X^{a}_{\mathrm{ave}} &\equiv& T^{-1} \int
_0^T X^{a}(t) \,\mathrm{d}t =
N^{-1} \sum_{n=0}^{N-1}
\frac{1}{4} \bigl(X^{a}_{n} + 2 X^{a}_{n+\halfs}
+ X^{a}_{n+1}\bigr).
\end{eqnarray*}

Due to H{\"o}lder's inequality,
\begin{eqnarray*}
\mathbb{E} \bigl[ \bigl\llVert X^{f}_{\mathrm{ave}} -
X^{a}_{\mathrm{ave}} \bigr\rrVert ^p \bigr] &\leq&
T^{-1} \int_0^T \mathbb{E} \bigl[
\bigl\llVert X^{f}(t) - X^{a}(t) \bigr\rrVert ^p
\bigr] \,\mathrm{d}t
\\
&\leq&\sup_{[0,T]} \mathbb{E} \bigl[ \bigl
\llVert X^{f}(t) - X^{a}(t) \bigr\rrVert ^p
\bigr]
\end{eqnarray*}
and similarly,
\[
\mathbb{E} \biggl[ \biggl\llVert \frac{1}{2}\bigl(X^{f}_{\mathrm{ave}}
+ X^{a}_{\mathrm{ave}}\bigr) - X^c_{\mathrm{ave}} \biggr
\rrVert ^p \biggr] \leq\sup_{[0,T]} \mathbb{E} \bigl[
\bigl\llVert \X^f(t) - X^c(t) \bigr\rrVert ^p
\bigr].
\]
Hence, if the Asian payoff is a smooth function of the average, then taking
$p = 2$ in Lemma \ref{lemmapayoff},
$p = 4$ in Corollary \ref{corodiff} and
$p = 2$ in Corollary \ref{coromain},
again gives a second order bound for the multilevel correction variance.

This analysis can be extended to include payoffs which are a smooth
function of a number of intermediate variables, each of which is a
linear functional of the path $x(t)$ of the form
\[
\int_0^T g^T(t) x(t) \mu(
\mathrm{d}t)
\]
for some vector function $g(t)$ and measure $\mu(\mathrm{d}t)$.
This includes weighted averages of $x(t)$ at a number of discrete
times, as well as continuously-weighted averages over the whole
time interval.

As with the European options, the analysis can also be extended to
payoffs which are Lipschitz functions of the average, and have
first and second derivatives which exist and are continuous and
uniformly bounded, except for a set of points~$K$ of zero measure.

%
\begin{assp} \label{aspLip2}
The payoff $P \in C(\mathbb{R}^d, \mathbb{R})$ has a uniform
Lipschitz bound,
so that there exists a constant $L$ such that
\[
\bigl\llvert P(x) - P(y) \bigr\rrvert \leq L \llvert x - y \rrvert \qquad
\forall x, y \in\mathbb{R}^d
\]
and the first and second derivatives exist, are continuous and have
uniform bound $L$ at all points $x \notin K$, where $K$ is a set of
zero measure, and there exists a constant $c$ such that the probability
of $x_{\mathrm{ave}}$ being within a neighbourhood of the set $K$
has the bound
\[
\mathbb{P} \Bigl( \min_{y\in K} \| x_{\mathrm{ave}} - y \| \leq
\varepsilon \Bigr) \leq c \varepsilon\qquad\forall \varepsilon> 0.
\]
\end{assp}

%
\begin{theorem} \label{thmLip2}
If the SDE satisfies the conditions of Assumption \ref{asgLip},
and the payoff satisfies Assumption \ref{aspLip2}, then
\[
\mathbb{E} \bigl[ \bigl(\tfrac{1}{2} \bigl(P\bigl(X_{\mathrm{ave}}^{f}
\bigr) + P\bigl(X_{\mathrm{ave}}^{a}\bigr)\bigr) - P
\bigl(X_{\mathrm{ave}}^c\bigr) \bigr)^2 \bigr] = o
\bigl(\Delta t^{3/2-\delta}\bigr)
\]
for any $\delta>0$.
\end{theorem}

\subsection{Nonasymptotic result}

The analysis above concerns the asymptotic behaviour of the multilevel
variance as $\Delta t\rightarrow0$. However, it is also worth noting
that since $X^{f}$ and $X^{a}$ have exactly the same distribution,
conditional on the\vadjust{\goodbreak} coarse path Brownian increments $\Delta W^c$, then
$P^{f} - P^c$ and $P^{a} - P^c$ are identically distributed,
and hence
%
\begin{eqnarray}\label{eqvaridentity}
\mathbb{V} \bigl[ \tfrac{1}{2}\bigl(P^{f} + P^{a}
\bigr) - P^c \bigr] &=& \mathbb{V} \bigl[ \tfrac{1}{2}
\bigl(P^{f} - P^c\bigr) + \tfrac{1}{2}
\bigl(P^{a} - P^c\bigr) \bigr]
\nonumber\\[-8pt]\\[-8pt]
&=&  \tfrac{1}{2}(1 + \rho) \mathbb{V}\bigl[P^{f} - P^c\bigr], \nonumber
\end{eqnarray}
where $\rho$ is the correlation between the $P^{f} - P^c$ and
$P^{a} - P^c$.
Thus, regardless of the size of the timestep, the variance of the
antithetic estimator cannot be larger than the variance of the standard
estimator, and could be significantly smaller if $\rho$~is negative.
What the asymptotic analysis shows is that $\rho\rightarrow-1$ as
$\Delta t \rightarrow0$.

\section{Numerical experiments}

In this section we present numerical tests in which we compare classical
Monte Carlo (MC), standard MLMC and antithetic MLMC estimators.
We consider the Clark--Cameron SDEs and Heston's stochastic volatility
model with both smooth and non-smooth payoffs. We will see that
in all cases the antithetic MLMC variance is significantly smaller than
the standard MLMC variance on all levels of approximation.

\subsection{Clark--Cameron SDEs}

The first set of results in Figure~\ref{figsin}
is for the Clark--Cameron SDEs with initial conditions
$x_{1}(0) = x_{2}(0) = 0$, final time $T=1$, and smooth payoff
$
P = \cos(x_1(T))$.

%
\begin{figure}

\includegraphics{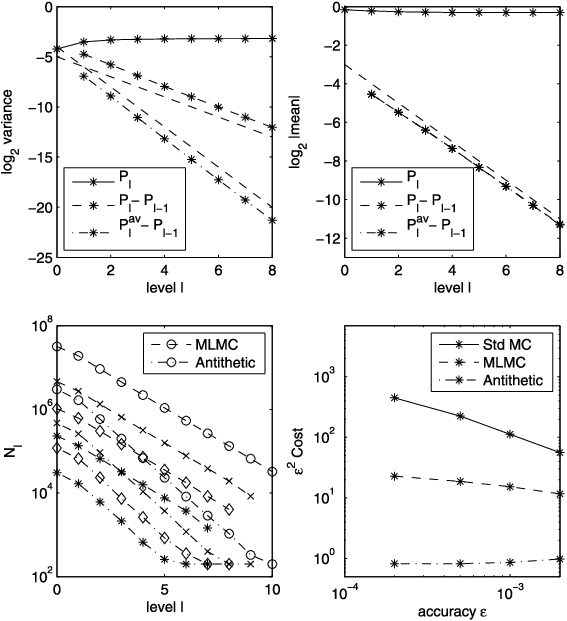}

\caption{Clark--Cameron SDEs with smooth payoff $P=\cos
(x_1(T))$.}\label{figsin}
\end{figure}

The top left plot shows the behaviour of the variance as a function
of the level of approximation, so that $\Delta t = 2^{- \ell}$.
These values were estimated using $10^6$ samples, so the sampling error
is very small.
The solid line is the variance of the standard Monte Carlo estimator
which varies very little with level. The dashed line is the usual MLMC
estimator $P^{f}_{\ell}-P^c_{\ell-1}$, and the accompanying reference
line with slope $-$1 confirms its expected first order convergence.
The dot-dash line is for the antithetic estimator
$\frac{1}{2}(P^{f}_{\ell} +P^{a}_{\ell}) - P^c_{\ell-1}$, and its
accompanying reference line with slope $-$2 confirms its second order
convergence. Note also that even on level $\ell=1$ in which the multilevel
estimator comes from the difference between simulations with 2 timesteps
(on level 1) and 1 timestep (on level 0), the antithetic estimator has a
variance which is roughly a factor 4 smaller than the standard MLMC
estimator.

The top right plot shows the mean value for the multilevel correction.
As expected the standard MLMC and antithetic MLMC estimator have exactly
the same expected value, and it converges at first order as indicated by
the reference line with slope $-$1.

The bottom right plot shows the dependence of the computational complexity
$C$ (defined as the total number of random numbers generated) as a function
of the desired accuracy $\epsilon$. Because of Theorem \ref{thcomplexity}
the plot is of $\epsilon^2 C$ versus $\epsilon$, because we expect
to see that
$\epsilon^2 C$ is only weakly dependent on $\epsilon$ for the
standard MLMC
and independent of $\epsilon$ for the antithetic MLMC. For the\vadjust{\goodbreak}
standard Monte
Carlo method, theory predicts that $\epsilon^{2} C$ should be
proportional to
the number of timesteps on the finest level, which in turn is roughly
proportional to $\epsilon^{-1}$ due to the first order weak
convergence order.
We see that computational complexity of the antithetic MLMC is much lower
than for the standard MLMC.

Further insight into the complexity cost is provided by the bottom left plot.
Each point in the bottom right complexity plot corresponds to a line in
the bottom left plot, showing the number of samples taken on each level of
the multilevel approximation. Lines with the same plotting symbol
correspond to the same desired accuracy $\epsilon$, with the upper
line being for
the MLMC estimator, and the lower line being for the antithetic estimator.

There are several points to note in this plot. The first is that for a
given accuracy, the number of samples on each level decays rapidly as
$\ell$
increases. This follows the prescription given in \cite{giles08} in
which the optimal number of samples on each level is proportional to
$\sqrt{V_l/C_l}$ where $V_l$ is the multilevel variance and $C_l$ is
the cost of a single sample on level $\ell$. The constant of proportionality
is chosen so that the overall variance
$
\sum_{\ell=0}^L N_\ell^{-1} V_\ell
$
is less than $\epsilon^2/2$. Because the antithetic
variance converges to zero more rapidly, the slope of the antithetic lines
is slightly greater than the slope of the standard MLMC lines.

The next point to note is that the lines with circular symbols (which
correspond to the tightest accuracy specification $\epsilon=10^{-4}$) extend
to level $\ell= 10$, while the other lines terminate at lower levels.
This is again following the prescription in \cite{giles08} in which the
mean square error is brought below $\epsilon^2$ by ensuring that the
square of the bias is also below $\epsilon^2/2$, like the total variance.
Using a simple heuristic to estimate the remaining discretisation bias,
because of the first order weak convergence, fewer approximation levels
are required when $\epsilon$ is larger.

The final observation to be made is that the antithetic line lies well below
the standard MLMC line for the same accuracy $\epsilon$. This is what produces
the overall computational savings shown in the bottom right plot. However,
on level $0$ the two are using exactly the same estimator, so why does the
antithetic estimator use fewer samples than the standard MLMC on level $0$?
The answer is that both have a variance budget of $\epsilon^2/2$ to
be spread
over all of the levels in the way which minimises the total computational
cost \cite{giles08}. In the standard MLMC case, this budget is spread
fairly evenly over the different levels, but in the antithetic case most
of the budget is allocated to level $0$ (because the estimator variance
decays so rapidly on the higher levels) and so fewer samples are required
on level $0$.

The next set of results in Figure~\ref{figCall} are for the same
Clark--Cameron SDE but with the Lipschitz payoff
\[
P = \max\bigl(x_1(T),0\bigr).
\]
The same comments as before apply to the plots in this figure. The
only difference is that the lower of the two reference lines in the top
left plot has slope $-$1.5, confirming that the multilevel variance
is $O(\Delta t^{3/2})$ rather than $O(\Delta t^{2})$ because of the
discontinuity in the first derivative of the payoff function.
Apart from that, the results are very similar with the antithetic estimator
have a much lower variance on all grid levels, and overall giving a much
reduced computational cost.

%
\begin{figure}

\includegraphics{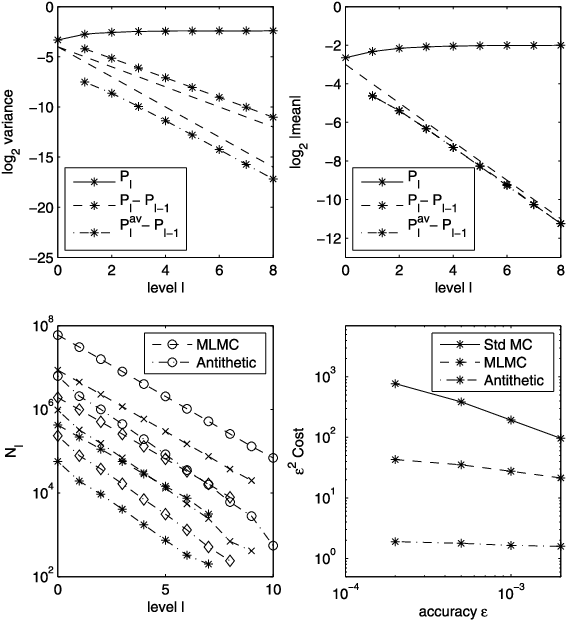}

\caption{Clark--Cameron SDEs with $P = \max(x_1(T),0)$.}\label{figCall}
\end{figure}

\subsection{Heston stochastic volatility model}

The Heston model \cite{heston93}, which is an asset price model
with stochastic volatility, is one of the most popular SDEs in finance
\begin{eqnarray*}
\mathrm{d} s(t) &=& r s(t) \,\mathrm{d}t + \sqrt{v(t)} s(t) \,\mathrm
{d}w_{1}(t), \qquad s(0)>0,
\\
\mathrm{d}v(t) &=& \kappa\bigl(\theta- v(t)\bigr) \,\mathrm{d}t + \sigma\sqrt
{v(t)} \,\mathrm{d}w_{2}(t), \qquad v(0)>0,
\end{eqnarray*}
where $\mathbb{E}[w_{1}(t)w_{2}(t)]=0$, $r>0$ and $2\kappa\theta\geq
\sigma^2$,
ensuring that the zero boundary is not attainable for the volatility process.
Due to the nonlinearity of the diffusion coefficient in the price process
$s(t)$ we work with log-Heston model
\begin{eqnarray*}
\mathrm{d}\log\bigl(s(t)\bigr) &=& \bigl(r - \tfrac{1}{2} v(t)\bigr) \,
\mathrm{d}t + \sqrt{v(t)} \,\mathrm{d}w_{1}(t),
\\
\mathrm{d}v(t) &=& \kappa\bigl(\theta- v(t)\bigr) \,\mathrm{d}t + \sigma
\sqrt{v(t)} \,\mathrm{d}w_{2}(t).
\end{eqnarray*}

%
\begin{figure}

\includegraphics{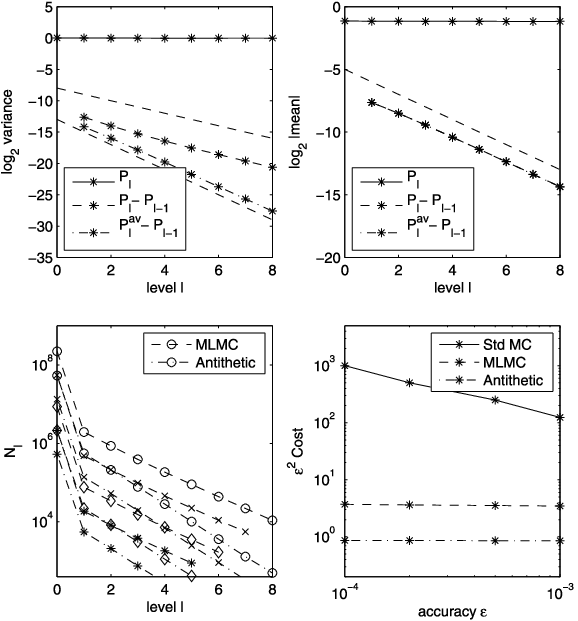}

\caption{Heston SDEs with $P=x(T)$.} \label{figHsmooth}
\end{figure}

Although the coefficients of the volatility process $\{v(t)\}_{t\geq
0}$ are
not Lipschitz continuous, and hence the assumptions imposed in the current
paper are not satisfied, the numerical tests show that the antithetic MLMC
performs very well. To approximate the volatility process we use a drift
implicit Milstein scheme that preserves the positivity of the original SDE,
and has a good strong convergence property recently established by
Neuenkirch and Szpruch in \cite{ns12}. Hence, the Milstein
scheme for Heston's stochastic volatility model with the L{\'e}vy area term
set to zero is given by
\begin{eqnarray*}
\log(S_{n+1}) &=& \log(S_{n}) + \bigl(r -
\tfrac{1}{2} V_{n}\bigr) \Delta t + \sqrt{V_{n}}
\Delta w_{1,n} + \tfrac{1}{4}\sigma\Delta w_{1,n}\Delta
w_{2,n},
\\
V_{n+1} &=& V_{n} + \kappa(\theta- V_{n+1})
\Delta t + \sigma\sqrt {V_{n}} \Delta w_{2,n} +
\tfrac{1}{2}\sigma^4\bigl(\Delta w_{2,n}^2
- \Delta t\bigr).
\end{eqnarray*}
For the simulation studies we choose $s_{0}=v_{0}=1$, $r=0.05$, $T=1$ and
$\kappa= 0.5$, $\theta= 0.9$, $\sigma=0.05$
in order to ensure the Feller boundary condition for the volatility process.

%
\begin{figure}

\includegraphics{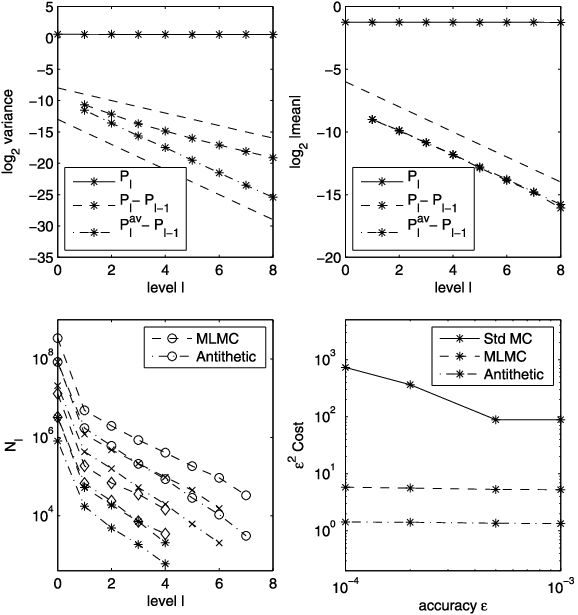}

\caption{Heston SDEs with $P = \max(s(T) - 1,0)$.} \label{figHcall}
\end{figure}

Figure~\ref{figHsmooth} presents our results for the smooth payoff
$P = x(T)$. The four plots have a similar structure to the results
of the Clark--Cameron application. The two reference lines in the top
left plot again have slopes $-$1 and $-$2, confirming that the antithetic
MLMC variance is $O(\Delta t^2)$, whereas the standard MLMC variance is
$O(\Delta t)$. The top right plot shows that the weak discretisation
error is again first order.

The bottom right plot shows that computational savings of the antithetic
MLMC compared to the standard MLMC are not as great as for the Clark--Cameron
application. The reason for this can be seen in the bottom left plot.
The multilevel variance on levels 1 and above is much smaller than the
variance on level 0, where both methods use the same estimator. Hence,
in both cases much of the computational effort is expended on the coarsest
level and so the benefits of the antithetic treatment are reduced.

The final results in Figure~\ref{figHcall} are for the same Heston
SDEs but with the call option payoff $P = \max(s(T)-1,0)$. The steeper
of the two reference lines
in the top left plot has a slope of $-$2, not the $-$1.5 used for the
Clark--Cameron case for the non-smooth payoff. This indicates that the
antithetic variance is $O(\Delta t^2)$, not the $O(\Delta t^{3/2})$ predicted
by the analysis. It is possible that there is indeed an $O(\Delta t^{3/2})$
component to the error, but that the corresponding coefficient is so small
that it does not become apparent until much smaller values of $\Delta t$.
Other than this, the results are very similar to the previous case.

\section{Conclusions}

In this paper we have constructed a new antithetic multilevel Monte Carlo
estimator for multi-dimensional SDEs, with a variance which is
$O(\Delta t^2)$
when the payoff function is smooth, and almost an $O(\Delta t^{3/2})$
when it
is Lipschitz and piecewise smooth. The algorithm is very easy to implement;
all that is required is to calculate a second fine path for which the odd
and even Brownian increments are swapped.\vadjust{\goodbreak}

In the European and Asian payoff cases considered in this paper, it
reduces the computational complexity for an $\epsilon$ root-mean-square
error to $O(\epsilon^{-2})$, compared to $O(\epsilon^{-2} (\log
{(1/ \epsilon)})^2)$ for
the multilevel method using the Euler--Maruyama discretisation, and
$O(\epsilon^{-3})$ for the standard Monte Carlo method. Furthermore, by
ensuring that the dominant computational effort is on the coarsest
levels (since $\beta>1$), it is now feasible to obtain further
improvements using quasi-Monte Carlo techniques \cite{gw09}.

In a future paper, we will extend the analysis to cover digital and
barrier options. The improvements from an extended version of the
antithetic treatment are then more substantial, improving the complexity
from $O(\epsilon^{-5/2})$ to approximately $O(\epsilon^{-2})$.

\begin{appendix}
\section*{Appendix: Proof of main results}\label{app}

\subsection{Proof of Lemma \texorpdfstring{\protect\ref{lemmadiff}}{4.6}}
Conditional on the Brownian increments $\Delta w$ for the coarse path $X^c$,
the Brownian increments for $X^{f}$ and $X^{a}$ have exactly the same
distribution, and therefore $X^{a}_n - X^{c}_n$ has exactly the
same distribution as $X^{f}_n - X^{c}_n$. Hence we obtain, using
inequality (\ref{eqineq}),
\begin{eqnarray*}
&& \mathbb{E} \Bigl[ \max_{0\leq n\leq N} \bigl\| X^{f}_n
- X^{a}_n \bigr\|^p \Bigr]
\\
&&\qquad \leq  2^{p-1}\Bigl( \mathbb{E} \Bigl[ \max_{0\leq n\leq N} \bigl\| X^{f}_n
- X^{c}_n \bigr\|^p \Bigr] + \mathbb{E} \Bigl[
\max_{0\leq n\leq N} \bigl\| X^{a}_n -
X^{c}_n \bigr\|^p \Bigr] \Bigr)
\\
&&\qquad =  2^p \mathbb{E} \Bigl[ \max_{0\leq n\leq N} \bigl\|
X^{f}_n - X^{c}_n
\bigr\|^p \Bigr]
\\
&&\qquad \leq 2^{2p-1} \Bigl( \mathbb{E} \Bigl[ \max_{0\leq n\leq N}
\bigl\| X^{f}_n - x(t_n) \bigr\|^p \Bigr] +
\mathbb{E} \Bigl[ \max_{0\leq n\leq N} \bigl\| X^{c}_n
- x(t_n) \bigr\|^p \Bigr] \Bigr).
\end{eqnarray*}
The desired result then follows from the strong convergence property in
Lemma~\ref{lemmastrong}.

\subsection{Proof of Lemma \texorpdfstring{\protect\ref{lemmaeqn1}}{4.7} and Corollary
\texorpdfstring{\protect\ref{coroeqn1}}{4.8}}
Combining the two equations in (\ref{eqeqn1}), and using the identity
\[
\Delta w_{j,n} \Delta w_{k,n} = (\delta w_{j,n} +
\delta w_{j,n+\halfs}) (\delta w_{k,n} + \delta w_{k,n+\halfs})
\]
together with the definition of $h_{ijk}$ in (\ref{eqh}) gives,
after considerable re-arrange\-ment,
\begin{eqnarray*}
X^{f}_{i,n+1} &=& X^{f}_{i,n} +
f_i\bigl(X^{f}_n\bigr) \Delta t + \sum
_{j=1}^D g_{ij}
\bigl(X^{f}_n\bigr) \Delta w_{j,n}
\\
&&{} + \sum_{j,k=1}^D h_{ijk}
\bigl(X^{f}_n\bigr) ( \Delta w_{j,n} \Delta
w_{k,n} - \Omega_{jk} \Delta t )
\\
&&{} - \sum_{j,k=1}^D h_{ijk}
\bigl(X^{f}_n\bigr) ( \delta w_{j,n} \delta
w_{k,n+\halfs} - \delta w_{k,n} \delta w_{j,n+\halfs} )
\\
&&{} + R_{i,n} + M^{(2)}_{i,n} + M^{(3)}_{i,n},
\end{eqnarray*}
where
\begin{eqnarray*}
R_{i,n} & = & \bigl(f_i\bigl(X^{f}_{n+\halfs}
\bigr)-f_i\bigl(X^{f}_{n}\bigr) \bigr) \Delta
t / 2,
\\
M^{(2)}_{i,n} & = & \sum_{j=1}^D
\Biggl(g_{ij}\bigl(X^{f}_{n+\halfs}
\bigr)-g_{ij}\bigl(X^{f}_{n}\bigr) - 2 \sum
_{k=1}^D h_{ijk}
\bigl(X^{f}_{n}\bigr) \delta w_{k,n} \Biggr)
\delta w_{j,n+\halfs},
\\
M^{(3)}_{i,n} & = & \sum_{j,k=1}^D
\bigl( h_{ijk}\bigl(X^{f}_{n+\halfs}\bigr) -
h_{ijk}\bigl(X^{f}_n\bigr) \bigr) ( \delta
w_{j,n+\halfs} \delta w_{k,n+\halfs} - \Omega_{jk} \Delta t / 2
).
\end{eqnarray*}

Considering $R_n$, a Taylor expansion gives
\begin{eqnarray*}
&& f_i\bigl(X^{f}_{n+\halfs}\bigr)-f_i
\bigl(X^{f}_{n}\bigr)
\\
&&\qquad = \sum_{j=1}^d \frac{\partial f_i}{\partial x_j}\bigl(X^{f}_{n}
\bigr) \bigl( X^{f}_{j,n+\halfs}- X^{f}_{j,n}
\bigr)
\\[-1pt]
&&\quad\qquad{} + \frac{1}{2} \sum_{j,k=1}^d
\frac{\partial^2 f_i}{\partial
x_j\,\partial x_k}(\xi_1) \bigl( X^{f}_{j,n+\halfs}-
X^{f}_{j,n} \bigr) \bigl( X^{f}_{k,n+\halfs}-
X^{f}_{k,n} \bigr)
\end{eqnarray*}
for some $\xi_1$ which lies on the line between $X^{f}_{n}$ and
$X^{f}_{n+\halfs}$.
Hence, $R_n$ can be split into two parts, $R_n = M^{(1)}_n + N_n$,
where
\[
M^{(1)}_{i,n} = \sum_{j=1}^d
\sum_{k=1}^D \frac{\partial
f_i}{\partial x_j}
\bigl(X^{f}_{n}\bigr) g_{jk}\bigl(X^{f}_n
\bigr) \delta w_{k,n} \Delta t/2,
\]
%
and
\begin{eqnarray*}
N_{i,n} &=& \sum_{j=1}^d
\frac{\partial f_i}{\partial x_j}\bigl(X^{f}_{n}\bigr) \Biggl(
f_j\bigl(X^{f}_n\bigr) \Delta t / 2
\\[-2pt]
&&\hspace*{64pt}{} + \sum
_{k,l=1}^D h_{jkl}
\bigl(X^{f}_n\bigr) ( \delta w_{k,n}\delta
w_{l,n} - \Omega_{kl}\Delta t /2 ) \Biggr) \Delta t / 2
\\
&&{} + \frac{1}{2} \sum_{j,k=1}^d
\frac{\partial^2 f_i}{\partial
x_j\,\partial x_k}(\xi_1) \bigl( X^{f}_{j,n+\halfs}-
X^{f}_{j,n} \bigr) \bigl( X^{f}_{k,n+\halfs}-
X^{f}_{k,n} \bigr) \Delta t / 2.
\end{eqnarray*}

Considering $M^{(2)}_n$, a Taylor expansion gives
\begin{eqnarray*}
&& g_{ij}\bigl(X^{f}_{n+\halfs}\bigr)-g_{ij}
\bigl(X^{f}_{n}\bigr)
\\
&&\qquad = \sum
_{k=1}^d \frac{\partial g_{ij}}{\partial x_k}\bigl(X^{f}_{n}
\bigr) \bigl( X^{f}_{k,n+\halfs}- X^{f}_{k,n}
\bigr)
\\
&&\quad\qquad{} + \frac{1}{2} \sum_{k,l=1}^d
\frac{\partial^2
g_{ij}}{\partial x_k\,\partial x_l}(\xi_2) \bigl( X^{f}_{k,n+\halfs}-
X^{f}_{k,n} \bigr) \bigl( X^{f}_{l,n+\halfs}-
X^{f}_{l,n} \bigr)
\end{eqnarray*}
for some $\xi_2$ on the line between $X^{f}_{n}$ and $X^{f}_{n+\halfs}$,
and therefore
\begin{eqnarray*}
M^{(2)}_{i,n} & = & \sum_{j=1}^D
\sum_{k=1}^d \frac{\partial g_{ij}}{\partial x_k}
\bigl(X^{f}_{n}\bigr) \Biggl( f_k
\bigl(X^{f}_{n}\bigr) \Delta t/2
\\[-4pt]
&&\hspace*{81pt}{} +\hspace*{-0.3pt} \sum
_{l,m=1}^D h_{klm}\hspace*{-0.3pt} \bigl(X^{f}_{n}
\bigr) ( \delta w_{l,n}\delta w_{m,n} - \Omega_{lm}
\Delta t /2 ) \Biggr) \delta w_{j,n+\halfs}
\\
&&{} + \frac{1}{2} \sum_{j=1}^D \sum
_{k,l=1}^d \frac{\partial^2 g_{ij}}{\partial x_k\,\partial x_l}(
\xi_2) \bigl( X^{f}_{k,n+\halfs}- X^{f}_{k,n}
\bigr) \bigl( X^{f}_{l,n+\halfs}- X^{f}_{l,n}
\bigr) \delta w_{j,n+\halfs}.
\end{eqnarray*}

Finally, considering $M^{(3)}_n$ we have
\begin{eqnarray*}
M^{(3)}_{i,n} & = & \sum_{j,k=1}^D
\bigl( h_{ijk}\bigl(X^{f}_{n+\halfs}\bigr) -
h_{ijk}\bigl(X^{f}_n\bigr) \bigr) ( \delta
w_{j,n+\halfs} \delta w_{k,n+\halfs} - \Omega_{jk} \Delta t / 2 )
\\
& = & \sum_{j,k=1}^D \sum
_{l=1}^d \frac{\partial h_{ijk}}{\partial
x_l}(\xi_3)
\bigl( X^{f}_{l,n+\halfs}- X^{f}_{l,n} \bigr) (
\delta w_{j,n+\halfs} \delta w_{k,n+\halfs} - \Omega_{jk} \Delta
t / 2 )
\end{eqnarray*}
for some $\xi_3$ on the line between $X^{f}_{n}$ and $X^{f}_{n+\halfs}$.

Setting $M^{f}_n \equiv M^{(1)}_n + M^{(2)}_n + M^{(3)}_n$,
it is clear that $\mathbb{E}[M^{f}_n | \mathcal{F}_n]=0$ since
$\delta w_n$ is independent of $X^{f}_n$, and
$\delta w_{n+\halfs}$ is independent of $X^{f}_n$ and $X^{f}_{n+\halfs}$.

All that remains is to bound the magnitude of
$\mathbb{E}[ \|M^{f}_n\|^p]$ and $\mathbb{E}[ \|N^{f}_n\|^p]$.
Looking at two of the terms in $M^{(2)}_{i,n}$, for example,
the uniform bound on the first derivatives of $g$, together with the
fact that $\delta w_{n+\halfs}$ is independent of both $X^{f}_n$ and
$\delta w_{n}$ leads to
\begin{eqnarray*}
&& \mathbb{E} \biggl[ \biggl\llvert \frac{\partial g_{ij}}{\partial x_k}\bigl(X^{f}_{n}
\bigr) h_{klm}\bigl(X^{f}_{n}\bigr) \delta
w_{l,n} \delta w_{m,n} \delta w_{j,n+\halfs} \biggr\rrvert
^p \biggr]
\\
&&\qquad\leq L^p \mathbb{E} \bigl[ \bigl\llvert h_{klm}
\bigl(X^{f}_{n}\bigr) \bigr\rrvert ^p \bigr]
\mathbb{E} \bigl[\llVert \delta w_{n} \rrVert
^{2p} \bigr] \mathbb{E} \bigl[ \llVert \delta w_{n+\halfs} \rrVert
^p \bigr]
\end{eqnarray*}
and the uniform bound on the second derivatives of $g$, together with the
fact that $\delta w_{n+\halfs}$ is independent of both $X^{f}_n$ and
$X^{f}_{n+\halfs}$ leads to
\begin{eqnarray*}
&& \mathbb{E} \biggl[ \biggl\llvert \frac{\partial^2 g_{ij}}{\partial x_k\,\partial x_l}(\xi_2) \bigl(
X^{f}_{k,n+\halfs}- X^{f}_{k,n} \bigr) \bigl(
X^{f}_{l,n+\halfs}- X^{f}_{l,n} \bigr) \delta
w_{j,n+\halfs} \biggr\rrvert ^p \biggr]
\\
&&\qquad\leq L^p \mathbb{E} \bigl[ \bigl\llVert X^{f}_{n+\halfs}-
X^{f}_{n} \bigr\rrVert ^{2p} \bigr] \mathbb{E}
\bigl[ \llVert \delta w_{n+\halfs} \rrVert ^p \bigr].
\end{eqnarray*}


Combining the uniform bound on
$\mathbb{E} [ \llvert  h_{ijk}(X^{f}_{n}) \rrvert ^{2p}  ]$
from Corollary \ref{corostrong}
with the bounds from Lemma \ref{lemonestep},
and standard results for the moments of Brownian increments,
gives the required $O(\Delta t^{3p/2})$ bound for each of the
two terms considered.

Deriving similar bounds for the other terms in
$M^{f}$ and $N^{f}$, and combining them using (\ref{eqineq}),
eventually gives the desired bounds for both
$\mathbb{E}[ \|M^{f}_n\|^p]$ and $\mathbb{E}[ \|N^{f}_n\|^p]$.

The proof is almost exactly the same for Corollary \ref{coroeqn1}.
The sign change in the second line of the equation in the statement
of the corollary is due to the swapping of the Brownian increments for
the first and second halves of the timestep.

\subsection{Proof of Lemma \texorpdfstring{\protect\ref{lemmaeqn2}}{4.9}}
Recalling that $\X^f = \frac{1}{2}(X^{f} + X^{a})$,
taking the average of the results from Lemma \ref{lemmaeqn1} and
Corollary \ref{coroeqn1} gives
\begin{eqnarray*}
\X^f_{i,n+1} &=& \X^f_{i,n} +
f_i\bigl(\X^f_n\bigr) \Delta t + \sum
_{j=1}^D g_{ij}\bigl(
\X^f_n\bigr) \Delta w_{j,n}
\\
&&{} + \sum
_{j,k=1}^D h_{ijk}\bigl(\X^f_n
\bigr) (\Delta w_{j,n}\Delta w_{k,n} - \Omega_{jk}
\Delta t )
\\
&&{} + \frac{1}{2} \bigl( M^{f}_{i,n} +
N^{f}_{i,n} + M^{a}_{i,n} +
N^{a}_{i,n} \bigr) + M^{(1)}_{i,n} +
M^{(2)}_{i,n} + M^{(3)}_{i,n} +
N^{(1)}_{i,n},
\end{eqnarray*}
where
\begin{eqnarray*}
N^{(1)}_{i,n} &=& \bigl( \tfrac{1}{2}\bigl(
f_i\bigl(X^{f}_n\bigr) + f_i
\bigl(X^{a}_n\bigr) \bigr) - f_i\bigl(
\X^f_n\bigr) \bigr) \Delta t,
\\
M^{(1)}_{i,n} &=& \sum_{j=1}^D
\biggl( \frac{1}{2}\bigl( g_{ij}\bigl(X^{f}_n
\bigr) + g_{ij}\bigl(X^{a}_n\bigr) \bigr) -
g_{ij}\bigl(\X^f_n\bigr) \biggr) \Delta
w_{j,n},
\\
M^{(2)}_{i,n} &=& \sum_{j,k=1}^D
\biggl( \frac{1}{2}\bigl( h_{ijk}\bigl(X^{f}_n
\bigr) + h_{ijk}\bigl(X^{a}_n\bigr) \bigr) -
h_{ijk}\bigl(\X^f_n\bigr) \biggr) ( \Delta
w_{j,n} \Delta w_{k,n} - \Omega_{jk} \Delta t ),
\\
M^{(3)}_{i,n} &=& \sum_{j,k=1}^D
\frac{1}{2} \bigl( h_{ijk}\bigl(X^{f}_n
\bigr) - h_{ijk}\bigl(X^{a}_n\bigr) \bigr) (
\delta w_{j,n} \delta w_{k,n+\halfs} - \delta w_{k,n} \delta
w_{j,n+\halfs} ).
\end{eqnarray*}
Setting
\begin{eqnarray*}
M_n &=& \tfrac{1}{2} \bigl( M^{f}_{n} +
M^{a}_{n} \bigr) + M^{(1)}_n +
M^{(2)}_n + M^{(3)}_n,\qquad
N_n = \tfrac{1}{2} \bigl( N^{f}_{n} +
N^{a}_{n} \bigr) + N^{(1)}_n,
\end{eqnarray*}
it is clear that $\mathbb{E}[M_n | \mathcal{F}_n]=0$, and
all that remains is to bound the magnitude of
$\mathbb{E}[ \|M_n\|^p]$ and $\mathbb{E}[ \|N_n\|^p]$.
By performing second order Taylor series expansions for $f(x)$ and $g(x)$,
and first order expansions for $h(x)$, all about $\X^f_n$, we obtain
\begin{eqnarray*}
N^{(1)}_{i,n} &=& \frac{1}{16} \sum
_{j,k=1}^d \biggl(\frac{\partial^2 f_i}{\partial x_j\,\partial x_k}(
\xi_1) + \frac{\partial^2 f_i}{\partial x_j\,\partial x_k}(\xi_2) \biggr)
\bigl(X^{f}_{j,n} - X^{a}_{j,n}\bigr)
\bigl(X^{f}_{k,n} - X^{a}_{k,n}\bigr)
\Delta t,
\\
M^{(1)}_{i,n} &=& \frac{1}{16} \sum
_{j=1}^D \sum_{k,l=1}^d
\biggl(\frac{\partial^2 g_{ij}}{\partial x_k\,\partial x_l}(\xi_3) + \frac{\partial^2 g_{ij}}{\partial x_k\,\partial x_l}(
\xi_4) \biggr)
\\
&&\hspace*{50pt}{}\times  \bigl(X^{f}_{k,n} -
X^{a}_{k,n}\bigr) \bigl(X^{f}_{l,n} -
X^{a}_{l,n}\bigr) \Delta w_{j,n},
\\
M^{(2)}_{i,n} &=& \frac{1}{4} \sum
_{j,k=1}^D \sum_{l=1}^d
\biggl(\frac{\partial h_{ijk}}{\partial x_l}(\xi_5) - \frac{\partial h_{ijk}}{\partial x_l}(
\xi_6) \biggr)
\\
&&\hspace*{45pt}{}\times  \bigl(X^{f}_{l,n} -
X^{a}_{l,n}\bigr) ( \Delta w_{j,n} \Delta
w_{k,n} - \Omega_{jk} \Delta t ),
\\
M^{(3)}_{i,n} &=& \frac{1}{4} \sum
_{j,k=1}^D \sum_{l=1}^d
\biggl(\frac{\partial h_{ijk}}{\partial x_l}(\xi_7) + \frac{\partial h_{ijk}}{\partial x_l}(
\xi_8) \biggr)
\\
&&\hspace*{45pt}{}\times \bigl(X^{f}_{l,n} -
X^{a}_{l,n}\bigr) ( \delta w_{j,n} \delta
w_{k,n+\halfs} - \delta w_{k,n} \delta w_{j,n+\halfs} )
\end{eqnarray*}
for some $\xi_1, \xi_3, \xi_5, \xi_7$ between $\X^f_n$ and $X^{f}_n$,
and $\xi_2, \xi_4, \xi_6, \xi_8$ between $\X^f_n$ and~$X^{a}_n$.

Using the same arguments as in the final part of the proof of Lemma
\ref{lemmaeqn1}, together with the bounds on
$\mathbb{E}[ \|M^{f}_n\|^p]$,
$\mathbb{E}[ \|M^{a}_n\|^p]$,
$\mathbb{E}[ \|N^{f}_n\|^p]$ and
$\mathbb{E}[ \|N^{a}_n\|^p]$,
leads to the required bounds for the moments of $M_n$ and $N_n$.

\subsection{Proof of Theorem \texorpdfstring{\protect\ref{thmmain}}{4.10}}
If we define
$\displaystyle S_n = \mathbb{E} [ \max_{m\leq n} \| \X^f_m
- X^c_m \|^p  ]$,
then inequality (\ref{eqineq}) gives
%
\begin{equation}
\label{eqmax} S_n \leq d^{p-1} \sum
_{i=1}^d \mathbb{E} \Bigl[ \max_{m\leq n}
\bigl\llvert \X^f_{i,m} - X^c_{i,m}
\bigr\rrvert ^p \Bigr].
\end{equation}

Taking the difference between the equation in Lemma \ref{lemmaeqn2} and
equation (\ref{eqreducedMilstein}), and summing over the first $m$ timesteps,
we obtain
\begin{eqnarray*}
\X^f_{i,m} - X^c_{i,m} &=& \sum
_{l=0}^{m-1} \bigl( f_i\bigl(
\X^f_{i,l}\bigr) - f_i\bigl(X^c_{i,l}
\bigr) \bigr) \Delta t
\\
&&{} +\sum_{l=0}^{m-1} \sum
_{j=1}^D \bigl( g_{ij}\bigl(\X
^f_{i,l}\bigr)-g_{ij}\bigl(X^c_{i,l}
\bigr) \bigr)\Delta w_{j,l}
\\
&&{}+\sum_{l=0}^{m-1} \sum
_{j,k=1}^D \bigl( h_{ijk}\bigl(\X
^f_{i,l}\bigr)-h_{ijk}\bigl(X^c_{i,l}
\bigr) \bigr) ( \Delta w_{j,l} \Delta w_{k,l} -
\Omega_{jk} \Delta t )
\\
&&{}+ \sum_{l=0}^{m-1} M_{i,l} +
\sum_{l=0}^{m-1} N_{i,l}
\end{eqnarray*}
and using inequality (\ref{eqineq}) again gives
%
\begin{eqnarray}
\label{eqmain} && \mathbb{E} \Bigl[ \max_{m\leq n} \bigl\llvert
\X^f_{i,m} - X^c_{i,m} \bigr\rrvert
^p \Bigr]\nonumber
\\
&&\qquad\leq 5^{p-1} \Biggl(\mathbb{E} \Biggl[ \max
_{m\leq n} \Biggl\llvert \sum_{l=0}^{m-1}
\bigl( f_i\bigl(\X^f_{i,l}\bigr) -
f_i\bigl(X^c_{i,l}\bigr) \bigr) \Delta t
\Biggr\rrvert ^p \Biggr]\nonumber
\\
&&\hspace*{59pt}{} + \mathbb{E} \Biggl[ \max
_{m\leq n} \Biggl\llvert \sum_{l=0}^{m-1}
\sum_{j=1}^D \bigl( g_{ij}\bigl(
\X ^f_{i,l}\bigr)-g_{ij}\bigl(X^c_{i,l}
\bigr) \bigr)\Delta w_{j,l} \Biggr\rrvert ^p \Biggr]
\\
&&\hspace*{59pt}{}  + \mathbb{E} \Biggl[ \max_{m\leq n} \Biggl\llvert \sum
_{l=0}^{m-1} \sum_{j,k=1}^D
\bigl( h_{ijk}\bigl(\X ^f_{i,l}
\bigr)-h_{ijk}\bigl(X^c_{i,l}\bigr)\bigr)\nonumber
\\
&&\hspace*{145pt}{} \times  ( \Delta
w_{j,l} \Delta w_{k,l} - \Omega_{jk} \Delta t )
\Biggr\rrvert ^p \Biggr]\nonumber
\\
&&\hspace*{79pt}{}  + \mathbb{E} \Biggl[ \max_{m\leq n}
\Biggl\llvert \sum_{l=0}^{m-1}
M_{i,l} \Biggr\rrvert ^p \Biggr]
 + \mathbb{E} \Biggl[ \max
_{m\leq n} \Biggl\llvert \sum_{l=0}^{m-1}
N_{i,l} \Biggr\rrvert ^p \Biggr] \Biggr).\nonumber
\end{eqnarray}

We now need to bound each of the five expectations on the right-hand side
of~(\ref{eqmain}). The last is the easiest, since
\[
\Biggl\llvert \sum_{l=0}^{m-1}
N_{i,l} \Biggr\rrvert ^p \leq m^{p-1} \sum
_{l=0}^{m-1} \llvert N_{i,l}
\rrvert ^p \leq n^{p-1} \sum_{l=0}^{n-1}
\llvert N_{i,l} \rrvert ^p
\]
and therefore
\[
\mathbb{E} \Biggl[ \max_{m\leq n} \Biggl\llvert \sum
_{l=0}^{m-1} N_{i,l} \Biggr\rrvert
^p \Biggr] \leq n^{p-1} \sum_{l=0}^{n-1}
\mathbb{E} \bigl[ \llvert N_{i,l} \rrvert ^p \bigr] \leq
c_1 (n \Delta t)^p \Delta t^p
\]
for some constant $c_1$ (which like other such constants in this
proof will depend on $p$, $L$ and $T$ but not on $\Delta t$) due to
Lemma \ref{lemmaeqn2}.

Similarly, there exists a constant $c_2$ such that
\begin{eqnarray*}
\mathbb{E} \Biggl[ \max_{m\leq n} \Biggl\llvert \sum
_{l=0}^{m-1} \bigl( f_i\bigl(
\X^f_{i,l}\bigr) - f_i\bigl(X^c_{i,l}
\bigr) \bigr) \Delta t \Biggr\rrvert ^p \Biggr] & \leq&
n^{p-1}\sum_{l=0}^{m-1} \mathbb{E}
\bigl[ \bigl\llvert f_i\bigl(\X^f_{i,l}\bigr) -
f_i\bigl(X^c_{i,l}\bigr) \bigr\rrvert
^p \bigr] \Delta t^p
\\
& \leq& c_2 (n \Delta t)^{p-1} \sum
_{m=0}^{n-1} S_m \Delta t
\end{eqnarray*}
with the second step being due to the uniform bound on the first
derivatives of $f$.

The other three expectations in (\ref{eqmain}) involve martingales,
and so we can use the discrete Burkholder--Davis--Gundy inequality
\cite{bdg72}.
Starting again with the easiest,
there are constants $c_3$, $c_4$ such that
\begin{eqnarray*}
\mathbb{E} \Biggl[ \max_{m\leq n} \Biggl\llvert \sum
_{l=0}^{m-1} M_{i,l} \Biggr\rrvert
^p \Biggr] &\leq& c_3 \mathbb{E} \Biggl[ \Biggl( \sum
_{m=0}^{n-1} ( M_{i,m}
)^2 \Biggr)^{p/2} \Biggr]
\\
&\leq& c_3
n^{p/2 - 1} \sum_{m=0}^{n-1} \mathbb{E}
\bigl[ | M_{i,m} |^p \bigr] \leq c_4 (n \Delta
t)^{p/2} \Delta t^p
\end{eqnarray*}
with the final step being due to Lemma \ref{lemmaeqn2}.

Similarly, there exists a constant $c_5$ such that
\begin{eqnarray*}
&& \mathbb{E} \Biggl[ \max_{m\leq n} \Biggl\llvert \sum
_{l=0}^{m-1} \sum
_{j=1}^D \bigl( g_{ij}\bigl(\X
^f_{i,l}\bigr)-g_{ij}\bigl(X^c_{i,l}
\bigr) \bigr)\Delta w_{j,l} \Biggr\rrvert ^p \Biggr]
\\
&&\qquad\leq c_5 n^{p/2-1} D^{p-1} \sum
_{m=0}^{n-1} \sum_{j=1}^D
\mathbb {E} \bigl[ \bigl\llvert \bigl( g_{ij}\bigl(
\X^f_{i,m}\bigr)-g_{ij}\bigl(X^c_{i,m}
\bigr) \bigr)\Delta w_{j,m} \bigr\rrvert ^p \bigr].
\end{eqnarray*}
Since $\Delta w_{j,m}$ is independent of both $\X^f_{i,m}$ and $X^c_{i,m}$,
it follows that
\[
\mathbb{E} \bigl[ \bigl\llvert \bigl( g_{ij}\bigl(
\X^f_{i,m}\bigr)-g_{ij}\bigl(X^c_{i,m}
\bigr) \bigr)\Delta w_{j,m} \bigr\rrvert ^p \bigr] =
\mathbb{E} \bigl[ \bigl\llvert g_{ij}\bigl(\X^f_{i,m}
\bigr)-g_{ij}\bigl(X^c_{i,m}\bigr) \bigr\rrvert
^p \bigr] \mathbb{E} \bigl[ \llvert \Delta w_{j,m} \rrvert
^p \bigr].
\]
Hence, because of the uniformly bounded first derivatives of $g$,
and standard results for the moments of Brownian increments, there
exists a constant $c_6$ such that
\[
\mathbb{E} \Biggl[ \max_{m\leq n} \Biggl\llvert \sum
_{l=0}^{m-1} \sum_{j=1}^D
\bigl( g_{ij}\bigl(\X ^f_{i,l}
\bigr)-g_{ij}\bigl(X^c_{i,l}\bigr) \bigr)\Delta
w_{j,l} \Biggr\rrvert ^p \Biggr] \leq c_6 (n
\Delta t)^{p/2-1} \sum_{m=0}^{n-1}
S_m \Delta t.
\]\eject

Finally, following the same approach, there exists a constant $c_7$
such that
\begin{eqnarray*}
&& \mathbb{E} \Biggl[ \max_{m\leq n} \Biggl\llvert \sum
_{l=0}^{m-1} \sum_{j,k=1}^D
\bigl( h_{ijk}\bigl(\X ^f_{i,l}
\bigr)-h_{ijk}\bigl(X^c_{i,l}\bigr) \bigr) ( \Delta
w_{j,l} \Delta w_{k,l} - \Omega_{jk} \Delta t )
\Biggr\rrvert ^p \Biggr]
\\
&&\qquad\leq c_5 (n \Delta t)^{p/2-1} \Delta
t^{p/2} \sum_{m=0}^{n-1}
S_m \Delta t.
\end{eqnarray*}

Since $n\Delta t \leq T$ in all of the above inequalities,
combining the above bounds for each term in (\ref{eqmain}),
and inserting these into (\ref{eqmax}), there then exists
a constant $c_8$ such that
\[
S_n \leq c_8 \Biggl( \Delta t^p + \sum
_{m=0}^{n-1} S_m \Delta t
\Biggr).
\]
The desired result is then obtained from a discrete Gr{\"o}nwall inequality.

\subsection{Proof of Lemma \texorpdfstring{\protect\ref{corodiff}}{4.11}}
The identity
$X^{f}_{n+\halfs} - X^{a}_{n+\halfs}
= (X^{f}_{n+\halfs} - X^{f}_n)
+ (X^{f}_n - X^{a}_n)
+ (X^{a}_n -X^{a}_{n+\halfs})$ gives
\begin{eqnarray*}
&& \bigl\llVert X^{f}_{n+\halfs} - X^{a}_{n+\halfs}
\bigr\rrVert ^p
\\
&&\qquad \leq 3^{p-1} \bigl( \bigl\llVert
X^{f}_{n+\halfs} - X^{f}_n \bigr\rrVert
^p + \bigl\llVert X^{f}_n -
X^{a}_n \bigr\rrVert ^p + \bigl\llVert
X^{a}_{n+\halfs} - X^{a}_n \bigr\rrVert
^p \bigr).
\end{eqnarray*}
It then follows from Lemmas \ref{lemonestep2} and
\ref{lemmadiff} that there exists a constant $K_p$,
independent of both $\Delta t$ and $n$, for which
\[
\mathbb{E} \bigl[ \bigl\llVert X^{f}_{n+\halfs} -
X^{a}_{n+\halfs} \bigr\rrVert ^p \bigr] \leq
K_p \Delta t^{p/2}.
\]

\subsection{Proof of Lemma \texorpdfstring{\protect\ref{coromain}}{4.12}}

Averaging the discrete equations for
$X^{f}_{n+\halfs}$ and $X^{a}_{n+\halfs}$,
and using the identities
$\delta w_n = \frac{1}{2} \Delta w_n +
\frac{1}{2}(\delta w_n-\delta w_{n+\halfs})$ and
$\delta w_{n+\halfs} = \frac{1}{2} \Delta w_n -
\frac{1}{2}(\delta w_n-\delta w_{n+\halfs})$,
gives
%
\begin{equation}
\X^f_{i,n+\halfs} = \X^f_{i,n} +
\frac{1}{2} f_i\bigl(\X^f_n\bigr)
\Delta t + \frac{1}{2} \sum_{j=1}^D
g_{ij}\bigl(\X^f_n\bigr) \Delta
w_{j,n} + N_{i,n}, \label{eqfinal1}
\end{equation}
where
\begin{eqnarray*}
N_{i,n} &=& \frac{1}{2} \biggl( \frac{1}{2}\bigl(
f_i\bigl(X^{f}_n\bigr) + f_i
\bigl(X^{a}_n\bigr)\bigr) - f_i\bigl(
\X^f_n\bigr) \biggr) \Delta t
\\
&&{} + \frac{1}{2} \sum_{j=1}^D
\biggl( \frac{1}{2}\bigl( g_{ij}\bigl(X^{f}_n
\bigr) + g_{ij}\bigl(X^{a}_n\bigr) \bigr) -
g_{ij}\bigl(\X^f_n\bigr) \biggr) \Delta
w_{j,n}
\\
&&{} + \frac{1}{4} \sum_{j=1}^D
\bigl( g_{ij}\bigl(X^{f}_n\bigr) -
g_{ij}\bigl(X^{a}_n\bigr) \bigr) ( \delta
w_{j,n} - \delta w_{j,n+\halfs})
\\[-2pt]
&&{} + \frac{1}{2} \sum_{j,k=1}^D
\biggl( h_{ijk}\bigl(X^{f}_n\bigr) \biggl(\delta
w_{j,n} \delta w_{k,n} - \frac{1}{2}\Omega_{jk}
\Delta t\biggr)
\\[-2pt]
&&\hspace*{50pt} + h_{ijk}\bigl(X^{a}_n\bigr)
\biggl(\delta w_{j,n+\halfs} \delta w_{k,n+\halfs}- \frac{1}{2}\Omega
_{jk}\Delta t\biggr) \biggr).
\end{eqnarray*}
Following the same method of analysis as in the proof of Lemma
\ref{lemmaeqn1} it can be proved that $\mathbb{E}[ |N_{i,n}|^p]$ has
an $O(\Delta t^p)$ bound.

Next, defining $X^c_{n+\halfs}$ to be the linear interpolant value
$\frac{1}{2}(X^c_n + X^c_{n+1})$, then the equation for $X^c_{n+1}$ yields
%
\begin{eqnarray}\label{eqfinal2}
X^c_{i,n+\halfs} &=& X^c_{i,n} +
\frac{1}{2} f_i\bigl(X^c_n\bigr)
\Delta t + \frac{1}{2} \sum_{j=1}^d
g_{ij}\bigl(X^c_n\bigr) \Delta
w_{j,n}
\nonumber\\[-9pt]\\[-9pt]
&&{} + \frac{1}{2} \sum_{j,k=1}^d
h_{ijk}\bigl(X^c_n\bigr) (\Delta
w_{j,n} \Delta w_{k,n} - \Omega_{jk} \Delta t).\nonumber
\end{eqnarray}

Subtracting (\ref{eqfinal2}) from (\ref{eqfinal1}) gives
\begin{eqnarray*}
\X^f_{i,n+\halfs}- X^c_{i,n+\halfs} &=&
\X^f_{i,n} - X^c_{i,n} +
\frac{1}{2} \bigl(f_i\bigl(\X^f_n
\bigr) - f_i\bigl(X^c_n\bigr) \bigr) \Delta
t
\\[-2pt]
&&{} + \frac{1}{2} \sum_{j=1}^d
\bigl( g_{ij}\bigl(\X^f_n\bigr)-
g_{ij}\bigl(X^c_n\bigr) \bigr) \Delta
w_{j,n}
\\[-2pt]
&&{} + N_{i,n} + \frac{1}{2} \sum_{j,k=1}^d
h_{ijk}\bigl(X^c_n\bigr) (\Delta
w_{j,n} \Delta w_{k,n} - \Omega_{jk} \Delta t).
\end{eqnarray*}
%
Using the bounds on $\mathbb{E}[ \| \X^f_n - X^c_n \|^p]$, the
bounded first derivatives of $f(x)$ and $g(x)$, the uniform bound
on $\mathbb{E}[ | h_{ijk}(X^c_n) |^p]$ and standard results for Brownian
increments, we can conclude that there exists a constant $K_p$,
independent of both $\Delta t$ and $n$, such that
such that
\[
\mathbb{E} \bigl[ \bigl\llVert \X^f_{i,n+\halfs}-
X^c_{i,n+\halfs} \bigr\rrVert ^p \bigr] \leq
K_p \Delta t^p.
\]
\end{appendix}



%

\printaddresses


\begin{thebibliography}{17}
\bibitem{bdg72}
%
\begin{binproceedings}[mr]
\bauthor{\bsnm{Burkholder},~\bfnm{D.~L.}\binits{D.~L.}},
\bauthor{\bsnm{Davis},~\bfnm{B.~J.}\binits{B.~J.}} \AND
\bauthor{\bsnm{Gundy},~\bfnm{R.~F.}\binits{R.~F.}}
(\byear{1972}).
\btitle{Integral inequalities for convex functions of operators on martingales}.
In \bbooktitle{Proceedings of the {S}ixth {B}erkeley {S}ymposium on
{M}athematical {S}tatistics and {P}robability ({U}niv. {C}alifornia,
{B}erkeley, {C}alif., 1970/1971)}
\bvolume{2}
\bpages{223--240}.
\bpublisher{Univ. California Press},
\blocation{Berkeley, CA}.
\bid{mr={0400380}}
\end{binproceedings}
%
\bptok{imsref}\vadjust{\goodbreak}%
\endbibitem

\bibitem{cc80}
%
\begin{bincollection}[mr]
\bauthor{\bsnm{Clark},~\bfnm{J.~M.~C.}\binits{J.~M.~C.}} \AND
\bauthor{\bsnm{Cameron},~\bfnm{R.~J.}\binits{R.~J.}}
(\byear{1980}).
\btitle{The maximum rate of convergence of discrete approximations for
stochastic differential equations}.
In \bbooktitle{Stochastic Differential Systems ({P}roc. {IFIP}--{WG}
7/1 {W}orking {C}onf., {V}ilnius, 1978)}.
\bseries{Lecture Notes in Control and Information Sci.}
\bvolume{25}
\bpages{162--171}.
\bpublisher{Springer},
\blocation{Berlin}.
\bid{mr={0609181}}
\end{bincollection}
%
\bptok{imsref}%
\endbibitem

\bibitem{dg95}
%
\begin{barticle}[mr]
\bauthor{\bsnm{Duffie},~\bfnm{Darrell}\binits{D.}} \AND
\bauthor{\bsnm{Glynn},~\bfnm{Peter}\binits{P.}}
(\byear{1995}).
\btitle{Efficient {M}onte {C}arlo simulation of security prices}.
\bjournal{Ann. Appl. Probab.}
\bvolume{5}
\bpages{897--905}.
\bid{issn={1050-5164}, mr={1384358}}
\end{barticle}
%
\bptok{imsref}%
\endbibitem

\bibitem{gl94}
%
\begin{barticle}[mr]
\bauthor{\bsnm{Gaines},~\bfnm{J.~G.}\binits{J.~G.}} \AND
\bauthor{\bsnm{Lyons},~\bfnm{T.~J.}\binits{T.~J.}}
(\byear{1994}).
\btitle{Random generation of stochastic area integrals}.
\bjournal{SIAM J. Appl. Math.}
\bvolume{54}
\bpages{1132--1146}.
\bid{doi={10.1137/S0036139992235706}, issn={0036-1399}, mr={1284705}}
\end{barticle}
%
\bptok{imsref}%
\endbibitem

\bibitem{giles08b}
%
\begin{bincollection}[mr]
\bauthor{\bsnm{Giles},~\bfnm{Mike}\binits{M.}}
(\byear{2008}).
\btitle{Improved multilevel {M}onte {C}arlo convergence using the
{M}ilstein scheme}.
In \bbooktitle{Monte {C}arlo and Quasi-{M}onte {C}arlo Methods 2006}
(\beditor{\bfnm{A.}\binits{A.}~\bsnm{Keller}},
\beditor{\bfnm{S.}\binits{S.}~\bsnm{Heinrich}} \AND
\beditor{\bfnm{H.}\binits{H.}~\bsnm{Niederreiter}}, eds.)
\bpages{343--358}.
\bpublisher{Springer},
\blocation{Berlin}.
\bid{doi={10.1007/978-3-540-74496-2_20}, mr={2479233}}
\end{bincollection}
%
\bptok{imsref}%
\endbibitem

\bibitem{giles08}
%
\begin{barticle}[mr]
\bauthor{\bsnm{Giles},~\bfnm{Michael~B.}\binits{M.~B.}}
(\byear{2008}).
\btitle{Multilevel {M}onte {C}arlo path simulation}.
\bjournal{Oper. Res.}
\bvolume{56}
\bpages{607--617}.
\bid{doi={10.1287/opre.1070.0496}, issn={0030-364X}, mr={2436856}}
\end{barticle}
%
\bptok{imsref}%
\endbibitem

\bibitem{ghm09}
%
\begin{barticle}[mr]
\bauthor{\bsnm{Giles},~\bfnm{Michael~B.}\binits{M.~B.}},
\bauthor{\bsnm{Higham},~\bfnm{Desmond~J.}\binits{D.~J.}} \AND
\bauthor{\bsnm{Mao},~\bfnm{Xuerong}\binits{X.}}
(\byear{2009}).
\btitle{Analysing multi-level {M}onte {C}arlo for options with
non-globally {L}ipschitz payoff}.
\bjournal{Finance Stoch.}
\bvolume{13}
\bpages{403--413}.
\bid{doi={10.1007/s00780-009-0092-1}, issn={0949-2984}, mr={2519838}}
\end{barticle}
%
\bptok{imsref}%
\endbibitem

\bibitem{giles2012stochastic}
%
\begin{barticle}[mr]
\bauthor{\bsnm{Giles},~\bfnm{Michael~B.}\binits{M.~B.}} \AND
\bauthor{\bsnm{Reisinger},~\bfnm{Christoph}\binits{C.}}
(\byear{2012}).
\btitle{Stochastic finite differences and multilevel {M}onte {C}arlo
for a class of {SPDE}s in finance}.
\bjournal{SIAM J. Financial Math.}
\bvolume{3}
\bpages{572--592}.
\bid{doi={10.1137/110841916}, issn={1945-497X}, mr={2968046}}
\end{barticle}
%
\bptok{imsref}%
\endbibitem

\bibitem{gw09}
%
\begin{bincollection}[mr]
\bauthor{\bsnm{Giles},~\bfnm{Michael~B.}\binits{M.~B.}} \AND
\bauthor{\bsnm{Waterhouse},~\bfnm{Benjamin~J.}\binits{B.~J.}}
(\byear{2009}).
\btitle{Multilevel quasi-{M}onte {C}arlo path simulation}.
In \bbooktitle{Advanced Financial Modelling}.
\bseries{Radon Ser. Comput. Appl. Math.}
\bvolume{8}
\bpages{165--181}.
\bpublisher{Walter de Gruyter},
\blocation{Berlin}.
\bid{doi={10.1515/9783110213140.165}, mr={2648461}}
\end{bincollection}
%
\bptok{imsref}%
\endbibitem

\bibitem{glasserman04}
%
\begin{bbook}[mr]
\bauthor{\bsnm{Glasserman},~\bfnm{Paul}\binits{P.}}
(\byear{2004}).
\btitle{Monte {C}arlo Methods in Financial Engineering}.
\bpublisher{Springer},
\blocation{New York}.
\bid{mr={1999614}}
\end{bbook}
%
\bptok{imsref}%
\endbibitem

\bibitem{heston93}
%
\begin{barticle}[auto:STB|2013/12/09|07:59:19]
\bauthor{\bsnm{Heston},~\bfnm{S.~I.}\binits{S.~I.}}
(\byear{1993}).
\btitle{A closed-form solution for options with stochastic volatility
with applications to bond and currency options}.
\bjournal{Rev. Financ. Stud.}
\bvolume{6}
\bpages{327--343}.
\end{barticle}
%
\bptok{imsref}%
\endbibitem

\bibitem{ks91}
%
\begin{bbook}[mr]
\bauthor{\bsnm{Karatzas},~\bfnm{Ioannis}\binits{I.}} \AND
\bauthor{\bsnm{Shreve},~\bfnm{Steven~E.}\binits{S.~E.}}
(\byear{1991}).
\btitle{Brownian Motion and Stochastic Calculus},
\bedition{2nd} ed.
\bseries{Graduate Texts in Mathematics}
\bvolume{113}.
\bpublisher{Springer},
\blocation{New York}.
\bid{doi={10.1007/978-1-4612-0949-2}, mr={1121940}}
\end{bbook}
%
\bptok{imsref}%
\endbibitem

\bibitem{kp92}
%
\begin{bbook}[mr]
\bauthor{\bsnm{Kloeden},~\bfnm{Peter~E.}\binits{P.~E.}} \AND
\bauthor{\bsnm{Platen},~\bfnm{Eckhard}\binits{E.}}
(\byear{1992}).
\btitle{Numerical Solution of Stochastic Differential Equations}.
\bpublisher{Springer},
\blocation{Berlin}.
\bid{mr={1214374}}
\end{bbook}
%
\bptok{imsref}%
\endbibitem

\bibitem{mt04}
%
\begin{bbook}[mr]
\bauthor{\bsnm{Milstein},~\bfnm{G.~N.}\binits{G.~N.}} \AND
\bauthor{\bsnm{Tretyakov},~\bfnm{M.~V.}\binits{M.~V.}}
(\byear{2004}).
\btitle{Stochastic Numerics for Mathematical Physics}.
\bseries{Scientific Computation}.
\bpublisher{Springer},
\blocation{Berlin}.
\bid{mr={2069903}}
\end{bbook}
%
\bptok{imsref}%
\endbibitem

\bibitem{muller02}
%
\begin{bmisc}[auto:STB|2013/12/09|07:59:19]
\bauthor{\bsnm{M{\"u}ller-Gronbach},~\bfnm{T.}\binits{T.}}
(\byear{2002}).
\bhowpublished{\textit{Strong approximation of systems of stochastic
differential equations}.
Habilitation thesis, TU, Darmstadt.}
\end{bmisc}
%
\bptok{imsref}%


\bibitem{ns12}
\begin{bmisc}[auto:STB|2014/01/06|10:16:28]
\bauthor{\bsnm{Neuenkirch},~\bfnm{Andreas}\binits{A.}} \AND
\bauthor{\bsnm{Szpruch},~\bfnm{Lukasz}\binits{L.}}
(\byear{2014}).
\bhowpublished{First order strong approximations of scalar SDEs with values in a domain.
\textit{Numer. Math.} To appear. arXiv preprint, available at \arxivurl{arXiv:1209.0390}.}
\end{bmisc}
\bptok{imsref}%
\endbibitem


\bibitem{rw01}
%
\begin{barticle}[mr]
\bauthor{\bsnm{Ryd{\'e}n},~\bfnm{Tobias}\binits{T.}} \AND
\bauthor{\bsnm{Wiktorsson},~\bfnm{Magnus}\binits{M.}}
(\byear{2001}).
\btitle{On the simulation of iterated {I}t\^o integrals}.
\bjournal{Stochastic Process. Appl.}
\bvolume{91}
\bpages{151--168}.
\bid{doi={10.1016/S0304-4149(00)00053-3}, issn={0304-4149}, mr={1807367}}
\end{barticle}
%
\bptok{imsref}%
\endbibitem

\bibitem{wiktorsson01}
%
\begin{barticle}[mr]
\bauthor{\bsnm{Wiktorsson},~\bfnm{Magnus}\binits{M.}}
(\byear{2001}).
\btitle{Joint characteristic function and simultaneous simulation of
iterated {I}t\^o integrals for multiple independent {B}rownian motions}.
\bjournal{Ann. Appl. Probab.}
\bvolume{11}
\bpages{470--487}.
\bid{doi={10.1214/aoap/1015345301}, issn={1050-5164}, mr={1843055}}
\end{barticle}
%
\bptok{imsref}%
\endbibitem

\end{thebibliography}
\end{document}